%% file: main.tex
\documentclass[11pt]{article}

\input{preamble}

\begin{document}
\maketitle
\begin{abstract}
Longitudinal observational data are often subject to informative missingness. For example, in healthcare research using electronic health records, covariates and outcomes are measured only when patients interact with the healthcare system. Whether these data are observed usually depends on the unobserved variable values themselves, introducing a common form of data missing not at random. We develop identification and estimation methods for the effects of longitudinal modified treatment policies under a self-censoring model which explicitly models this dependence. We consider effects of treatments that can be modified only at measurement times (e.g., interactions with the heathcare system). We study two settings: one in which the observed history suffices for treatment confounding control, and another in which confounder control requires covariate values that were not measured at missed visits. In both settings, we combine longitudinal causal identification with shadow-variable conditional moment restrictions to obtain recursive identification formulas for the counterfactual outcome mean. We derive gradients of the pathwise derivatives and second-order remainder expansions, and use these results to construct estimators that accommodate machine learning for nuisance-function estimation. Under suitable product-rate conditions and cross-fitting, the estimators are asymptotically linear and permit root-n inference, including simultaneous inference across outcome times. We also describe conditional-moment learning and ensemble approaches for estimating the nuisance functions required by our approach, together with a pool-and-smooth strategy that borrows information across sparsely measured outcomes.
\end{abstract}

\input{observed_main}
\input{full_main}
\input{estimation_main}
\input{conclusion}

\bibliographystyle{plainnat}
\bibliography{refs}

\appendix
\input{observed_proofs}
\input{full_proofs}

\end{document}

%% file: preamble.tex
\usepackage[margin=1in]{geometry}
\usepackage{tikz}
\usepackage{here}
\usepackage{amsmath}
\usepackage{amsthm}
\usepackage{amssymb}
\usepackage{xcolor}
\usepackage{authblk}
\usepackage[round]{natbib}
\usepackage[colorlinks,citecolor=blue,urlcolor=blue]{hyperref}
\usepackage{cleveref}
\usepackage[expansion=false]{microtype}

\renewcommand{\underbar}[1]{\underline{#1}}
\usetikzlibrary{arrows.meta}

\renewcommand{\P}{\mathsf{P}}
\newcommand{\p}{\mathsf{p}}
 
\newcommand{\E}{\mathsf{E}}
\newtheorem{example}{Example}

\crefname{example}{Example}{Examples}
\Crefname{example}{Example}{Examples}
\newtheorem{assumption}{Assumption} 
\newtheorem{proposition}{Proposition}
\newtheorem{corollary}{Corollary} 
\newtheorem{theorem}{Theorem} 
\newtheorem{lemma}{Lemma}
\newcommand{\indep}{\mathrel{\perp\!\!\!\perp}}

\crefname{assumption}{Assumption}{Assumptions} 
\Crefname{assumption}{Assumption}{Assumptions}

\DeclareMathOperator{\supp}{supp}

\author[1]{Iv\'an D\'iaz}
\affil[1]{\small Division of Biostatistics, Department of Population
      Health, New York University Grossman School of Medicine, New York, NY, USA}

\title{Longitudinal causal inference under informative missingness using the self-censoring model.}

\date{October 6, 2026}

%% file: observed_main.tex
\section{Introduction}
Longitudinal observational data are increasingly used to study the effects of interventions in settings where randomized trials are unavailable or infeasible. Examples include prospective cohort studies, electronic health records, claims databases, registries, and other forms of real-world data. For instance, electronic health records are increasingly regarded as an important source of evidence on treatment effectiveness and safety in routine clinical practice, with their use receiving growing attention from clinical researchers and regulatory agencies \citep{abbasi2025promise,fda2024rwd}. 

Since \cite{robins1986new} laid the foundation for the identification and estimation of effects of time-varying treatments, a large literature has developed general methods for estimating causal effects from longitudinal data. Much of this literature, however, is developed for data measured on a fixed time grid under the assumption that the relevant treatment, covariate, and outcome histories are fully observed, or that missingness is ignorable conditional on the observed history. These assumptions are often violated in the longitudinal data sources that motivate modern causal analyses, leading to important biases. In electronic health records, for example, covariates and outcomes are typically observed only when patients interact with the health system. The reasons for these interactions are themselves often due to underlying health status or treatment needs, creating a problematic form of nonignorable data missing.  A general framework that simultaneously accommodates time-varying treatments, longitudinal causal estimands, and visit or measurement processes that depend directly on the variables of interest is missing from the literature. 

A solution posed in the literature addresses informative measurement by intervening on the visit
process itself by defining joint regimes on the treatment and visit process under a
no-direct-effect assumption \citep[e.g.,][]{neugebauer2017identification,
neugebauer2020exploiting, caniglia2019emulating}. These approaches address the missing-data problem by design of the causal question: under the intervened regime
the relevant variables are observed by construction. Encounters with the healthcare system, however, do not only represent measurement events: interventions that alter
healthcare utilization also alter downstream outcomes through treatment pathways that need not be of substantive causal interest. Furthermore, the interventions of interest (e.g., prescription medications, surgeries, etc.) often can only occur at such visits, making the visit process central to the definition of the treatment regimens. We therefore target the effect of interventions on
the treatments/exposures of interest alone, leaving the visit process at its natural value.
This choice defines an estimand of direct scientific interest, but it means that missingness must be addressed through means other than interventions on the visit process. Other related works do not explicitly describe an intervention on the monitoring process, but assume missingness is completely explained by measured data \citep{lin2004analysis, buzkova2009semiparametric,
pullenayegum2016longitudinal, coulombe2021weighted}. 

We address this gap in the literature by bringing together two strands of the causal inference literature. The first is the literature on longitudinal causal inference, beginning with the longitudinal g-computation formula for time-varying treatments \citep{robins1986new} and continuing through  developments enabling the use of machine learning regression to estimate nuisance parameters \citep{van2011targeted, luedtke2017sequential, rotnitzky2017multiply} and estimation of general effects based on longitudinal modified treatment policies \citep{young2014identification, richardson2013single, diaz2023nonparametric, hoffman2024studying}. The second is the literature on self-censoring and shadow-variable models for nonignorable missingness \citep{miao2015identification, miao2016varieties, miao2024identification, sun2018semiparametric, li2023non, li2023self}, which allow identification in situations where the missingness indicator depends directly on the unobserved time-varying variables. Instead of blocking all backdoor paths, identification is achieved through exclusion restrictions on a so-called \textit{shadow variable}, assumed to affect the missingness indicator only through its effect on the covariates and outcomes. 

Our results on the self-censoring  model are most closely related to the work of \citet{miao2024identification}, \citet{chen2023causal}, \citet{zhang2023proximal}, and \citet{li2023non}. \citet{miao2024identification} study identification and efficient estimation of the distribution of an outcome using a shadow variable, giving nonparametric identification conditions for the full-data law, semiparametric and doubly robust estimators, and the corresponding efficiency theory. \citet{chen2023causal} study causal effect identification when the outcome may cause its own missingness, or self-censoring, and show that, under conditions, the treatment can serve as a shadow variable leading to an inverse-probability-weighted estimator. \citet{zhang2023proximal} develop proximal methods for nonignorable missing outcomes without requiring uniqueness of the bridge function, whereas \citet{li2023non} derive conditions under which a mean functional is identified using a shadow variable even when the full-data distribution is not identified.

The present paper differs from this prior work in several respects. First, our target is a longitudinal causal parameter under a time-varying modified treatment policy, rather than a marginal mean or a single-time-point causal effect as in previous self-censoring work.  Second, we  leverage recent pool-and-smooth methods \citep{susmann2025curve} to address high-sparsity in the measurement process. Third, prior shadow-variable and self-censoring work often uses odds-ratio or missingness-odds parameterizations of bridge functions. In contrast, we define the true bridge as an inverse missingness probability and characterize arbitrary bridge solutions directly through a conditional moment equation. This avoids the need for preliminary estimation of missingness probabilities and the subsequent construction of potentially unstable odds, and places bridge estimation in the standard conditional-moment framework. Fourth, we avoid parametric modeling of the bridge function, instead working with nonparametric linear operators and their adjoints. Finally, we characterize pathwise differentiability and semi-parametric estimation for causal effects under a non-parametric longitudinal model that is restricted only by the shadow variable assumptions. 

We consider the effects of treatments that can be \textit{assigned or modified} only during encounters. The leading example is the effect of prescription medications in electronic health records, where prescriptions can generally be initiated, discontinued, or changed only when a patient interacts with the healthcare system. With biomedical applications in mind, we develop methods for two distinct types of treatments. In the first setting, we consider treatments for which adjustment for the observed history, including the prior visit pattern, is sufficient for treatment confounding control. An example is a medication prescription choice driven primarily by clinical information available to the provider at the encounter, such as selection of an antibiotic based on the documented diagnosis, culture results, allergies, and laboratory measurements. In the second setting, we consider treatments that require adjustment for the full data history, including variables that may not have been measured. Leading examples include settings in which patient preferences play an important role in treatment selection, for example when a patient is deciding between gastric bypass surgery and GLP-1 medications. For these exposures, adequate confounding control may require adjustment for the full latent covariate history, including covariate values that were not measured at missed visits.

Our approach proceeds in three steps. First, we formulate a longitudinal causal model with a covariate-dependent visit process reflecting our non-ignorable missing data structure. Second, we give identification results for the counterfactual mean under a longitudinal modified treatment policy by combining standard longitudinal causal identification arguments with a shadow-variable bridge representation for the missing outcome/covariate process. Third, we derive the corresponding semiparametric estimation theory. We characterize the tangent space induced by the bridge equation, obtain gradients of the pathwise derivatives, and use these gradients to construct targeted/double machine learning estimators. The construction of these estimators relies on standard theory, including a second-order a von Mises expansion, that provides the conditions under which they satisfy desirable properties such as asymptotic normality. Estimation requires solving Fredholm integral equations for the bridge and adjoint functions, for which we use modern conditional moment and minimax learning approaches \citep{muandet2020kernel, lewis2018adversarial, kallus2021causal, kompa2022deep}. We propose to estimate the causal density ratios induced by the modified treatment policy using probabilistic classification, Riesz learning, or more general Bregman density-ratio losses \citep{qin1998inferences, chernozhukov2018double, hines2025learning}. 

\section{Observed data and causal model}\label{sec:model}
Assume that at each time \(t\in\{1,\ldots,\tau+1\}\), there is a vector \(C_t\) containing the time-varying variables of interest, including biomarkers, comorbidities, potential confounders, and outcomes. These variables are defined continuously at every time point, but are observed at discrete times, only when a measurement or clinical encounter occurs. At each time point, there is also a vector \(A_t\) containing all treatment variables, such as medication prescriptions, therapies, and other clinical interventions. We are interested in the effects of interventions on selected components of the longitudinal process on an outcome at time \(t+1\in\{2,\ldots,\tau+1\}\), defined as \(Y_{t+1}:=h(C_{t+1})\) for some function \(h\), such as a coordinate projection.

We assume that the baseline variables \(C_1\) are always observed. This assumption is natural in our motivating applications, where study entry or cohort eligibility is defined by an observation, such as an encounter with the health system, and baseline demographic, clinical, and treatment-history variables are recorded at or before that encounter. By contrast, follow-up covariates and outcomes may be observed only when patients subsequently self select for a visit. Let \(R_t\) denote an indicator that a visit or measurement occurs at time \(t>1\), so that \(C_t\) is observed only if \(R_t=1\), and set \(R_1\equiv1\). Importantly, \(C_t\) exists regardless of whether \(R_t=1\); the visit process determines whether it is observed. We assume the following structural causal model:
\begin{align*}
    C_t & = f_{C,t}(\bar A_{t-1}, \bar R_{t-1}, \bar C_{t-1}, U_{C,t}),\\
    R_t & = f_{R,t}(\bar A_{t-1}, \bar R_{t-1}, \bar C_t, U_{R,t}),\\
    A_t & = f_{A,t}(\bar A_{t-1}, \bar R_t, \bar C_t, U_{A,t}),
\end{align*}
where $\bar X_t:=(X_1,\ldots, X_t)$ is the history of the variable $X$ up to time $t$, with $\bar X_0$ defined as the empty history. Here $U_{X,t}$ denotes all unobserved variables that, together with the observed parents of $X_t$, determines its value through the function $f_{X,t}$. Exclusion restrictions and conditional independence assumptions on this model will be clarified in each section below. We will use $\underline X_t:=(X_t, \ldots, X_{\tau+1})$ to denote the future of a random variable including $t$. Define $L_1=C_1$, and for $t>1$ define
\begin{equation*}
    L_t=\begin{cases}
        C_t&\text{ if } R_t=1\\
        \zeta &\text{ if } R_t=0,
    \end{cases}
\end{equation*}
where $\zeta\not\in \mathcal C_t$ is an arbitrary constant used to indicate that the random variable $C_t$ is unobserved in the event $R_t=0$.  

We develop methods for two distinct classes of causal questions. In \S\ref{sec:visit-dependent}, we consider interventions on the treatment history \((A_1,\ldots,A_t)\) in settings where adjustment for the observed history, including the prior visit pattern, is sufficient for confounding control. These methods are intended for treatments whose values can be initiated or modified only when a visit occurs, such as prescription medications in pharmacoepidemiologic research. In \S\ref{sec:full-history}, we consider the same class of interventions on \((A_1,\ldots,A_t)\), but in settings where adequate confounding control requires adjustment for the full data history, including covariate values that may not have been measured at missed visits. These methods are intended for treatment decisions that may depend on information not fully captured in the observed record, such as patient preferences or evolving latent health status.

We will use \(H_t=(L_1,A_1,\ldots,L_{t-1},A_{t-1},L_t)\) to denote the history of observed data up to right before \(A_t\), and \(F_t=(C_1,A_1,\ldots,C_{t-1},A_{t-1},C_t)\) to denote the corresponding full-data history. Notice that \(L_t\) encodes \(R_t\) and therefore the latter is omitted from the notation in \(H_t\). Methods based on the full-data history will require additional restrictions relating the visit process to the latent covariate process, formalized in \Cref{ass:fh-measurement-exchange}; we elaborate on these assumptions in \S\ref{sec:full-history}. The causal models used in this work are illustrated in Figure~\ref{fig:dag} for \(\tau=2\).



\section{Causal effects under observed-history adjustment}\label{sec:visit-dependent}

Throughout the paper, we assume that the vector $A_t$ contains only treatments that are initiated, discontinued, or modified at a clinical encounter. Prescription medications are the leading example: if no visit occurs at time $t$, then the treatments cannot be modified at that time, except through continuation of a prior prescription. In this section we assume that treatment decisions are based solely on the information recorded in the database. We formalize this in the following assumption. 

\begin{assumption}[Observed-history treatment mechanism]\label{ass:visit-dep} Assume that the structural equation for $A_t$ satisfies the restriction
\[
    A_t =
    \begin{cases}
        f_{A_t}^{(1)}(\bar A_{t-1}, \bar R_{t-1}, \bar L_t, U_{A,t}),
        & R_t=1,\\
        \alpha_t(H_t), & R_t=0,
    \end{cases}
\]
where \(\alpha_t\) is a potentially known deterministic continuation rule that determines treatment status at a no-visit time \(t\) from the previous treatments. 
\end{assumption}
In this assumption, \(\bar A_{t-1}\) is understood to include the treatment
information needed to apply the continuation rule, such as prescription date, dose, days
supplied, and refill information when relevant. For example,
\(\alpha_t\) may encode whether a prior prescription remains active. In the absence of a visit, this assumption states that treatment status evolves deterministically according to the previously observed treatment history. 

The restriction to treatments that can only be modified at visit times has two
important implications for identification of causal effects. First, the treatment assigned at time \(t\) can often be assumed to depend only on the
recorded history, including \(L_s\) when \(R_s=1\), but not on components of
\(C_s\) that are unrecorded because \(R_s=0\), for $s<t$. For example, we can assume that a clinician does not
directly base a treatment decision at time \(t\) on a blood pressure value at a prior time $s<t$ that was not
measured or recorded in the electronic health record because the patient did
not have a visit at that time. Therefore, under the visit-dependent treatment
assumption, unrecorded components of \(C_s\) at times with \(R_s=0\) do not confound
the relation between the current treatment \(A_t\) and subsequent variables, although they may still affect future covariates and outcomes. Second, for visit-dependent treatments, \(A_t\) may be known even when \(R_t=0\),
provided that exposure status in the absence of a visit is determined by a
prespecified rule based on the previously observed treatment history. For
example, this rule may encode whether a prior prescription remains active or has
expired. \Cref{ass:visit-dep} is not strictly required for the technical results in this section, but it is important in applications for defining interventions and causal effects that are practically meaningful (see Example \ref{ex:grace}). A version of \Cref{ass:visit-dep} will be required for identification in \S\ref{sec:full-history}.

Assume the causal parameter is defined in terms of a modified treatment policy \citep{diaz2023nonparametric}. Specifically, at time $t=1$, consider an intervention that assigns treatment according to $d_1(A_1, L_1)$. This generates counterfactual covariates, visit indicators, and treatments at time $t=2$, denoted by $C_2(d_1)$, $R_2(d_1)$, and $A_2(d_1)$, respectively. At time $t=2$, the modified treatment policy assigns treatment according to $d_2\{A_2(d_1), L_2(d_1)\}$. This subsequently generates $C_3(\bar d_2)$ and $R_3(\bar d_2)$. The process is repeated across time points, and, for simplicity, we denote $
A_t^d = d_t\{A_t(\bar d_{t-1}), H_t(\bar d_{t-1})\}$,
where $H_t(\bar d_{t-1}) =
\{L_1, A_1^d, L_2(d_1), \ldots, A_{t-1}^d, L_t(\bar d_{t-1})\}$
denotes the history that would have been observed under the intervention through time $t-1$. Here, we emphasize that the function $d_t$ need not modify all components of the treatment vector $A_t$. For instance, if one is interested in the effects of antidepressants on weight gain, then the function $d_t$ can be defined such that it only modifies the coordinate of $A_t$ corresponding to the antidepressant of interest. This flexibility will be important in constructing data structures for which some of the identification assumptions below are feasible. 

In addition, the intervention must be specified in a way that respects the visit process.
For visit-dependent treatments, the policy may modify treatment only at time points with a naturally occurring visit. When no visit occurs, treatment evolves according to the no-visit continuation rule and the natural value of treatment. We give two examples below.

\begin{example}[Prescription medication with grace period]\label{ex:grace}
    Consider a medication prescription with a grace period \citep{mcgrath2025time} of one time point. At time $t=1$, patients are assigned to initiate the medication, that is, $d_1(A_1, H_1)=1$. This generates a counterfactual covariate vector $C_2(d_1)$, a counterfactual visit indicator $R_2(d_1)$, as well as a natural value of treatment $A_2(d_1)$. At time point $t=2$, since there is a grace period of one time point, the patient is allowed to take the medication level they would have naturally taken, i.e., $d_2\{A_2(d_1), H_2(d_1)\} = A_2(d_1)$. At time point $t\geq 3$, if the patient is at the end of a grace period and would have had a visit under the intervention history through time $t-1$, they are assigned the prescription; otherwise, they are allowed to take what they would have naturally taken. That is:
\[
 d_t\{A_t(\bar d_{t-1}), H_t(\bar d_{t-1})\}
 =
\begin{cases}
    1, & \text{if } R_t(\bar d_{t-1}) = 1 \text{ and } A_{t-1}^d=0,\\
    A_t(\bar d_{t-1}), & \text{otherwise.}
\end{cases}
\]
\end{example}

\begin{example}[Treatment intensification]
Consider an antihypertensive medication whose dose can be increased only at a clinical encounter. Let $A_t\in\{0,1,\ldots,a_{\max}\}$ denote the dose level at time $t$, where larger values correspond to higher treatment intensity, and let $S_t\in C_t$ denote systolic blood pressure. Suppose the intervention initiates treatment at time $t=1$, so that $d_1(A_1,H_1)=1$. At later time points, the intervention increases the dose by one level if, under the intervention history, the patient has a visit and the measured systolic blood pressure is above 140 mmHg; otherwise, the patient receives the treatment value they would have naturally received. Formally, for $t\geq 2$,
\[
 d_t\{A_t(\bar d_{t-1}),H_t(\bar d_{t-1})\}
 =
 \begin{cases}
 \min\{A_t(\bar d_{t-1})+1,a_{\max}\},
 & \text{if } R_t(\bar d_{t-1})=1 \text{ and } S_t(\bar d_{t-1})>140,\\
 A_t(\bar d_{t-1}),
 & \text{otherwise.}
 \end{cases}
\]
\end{example}
Assume that we are interested in the counterfactual outcome mean $\theta_{t+1}(\bar d_t) = \E[Y_{t+1}(\bar d_t)]$ for some fixed time $t\leq \tau$, or in the curve $\{\theta_{t+1}(\bar d_t):t\geq 1\}$. In this section we develop methods under the assumption that the observed history $H_s:s\leq t$ provides sufficient confounder control:

\begin{assumption}[Sequential randomization conditional on observed history]\label{ass:nuc}
For every $s\leq t$, assume that
\[
U_{A,s}
\indep
\left(
\underbar U_{C,s+1},
\underbar U_{R,s+1},
\underbar U_{A,s+1}
\right)
\mid H_s.
\]
\end{assumption}
This assumption corresponds directly to the assumption that all common causes of $A_s$ and all subsequent variables are measured and adjusted for. We also require that the intervention of interest is well supported in the data.
\begin{assumption}[Supported intervention]\label{ass:support}
At baseline, if $(a_1,l_1)\in\supp(A_1,L_1)$, then $(a_1^d,l_1)\in\supp(A_1,L_1)$. For each $s=2,\ldots,t$, if
$(a_s,l_s)\in \supp\{A_s,L_s\mid A_{s-1}=a_{s-1}^d,H_{s-1}=h^d_{s-1}\}$,
then
$(a_s^d,l_s)\in \supp\{A_s,L_s\mid A_{s-1}=a_{s-1}^d,H_{s-1}=h^d_{s-1}\}$.
\end{assumption}
Under these conditions, standard arguments for the identification of effects with time-varying confounders \citep[e.g.,][]{richardson2013single, young2014identification, diaz2023nonparametric} show that the identifying functional can be written through sequential regressions. First define the terminal regression
\[
m_{t+1,t}(a_t,h_t)
:=
\E(Y_{t+1}\mid A_t=a_t,H_t=h_t).
\]
Then, for $s=t-1,\ldots,1$, define recursively
\[
m_{t+1,s}(a_s,h_s)
:=
\E\!\left[
 m_{t+1,s+1}\{d_{s+1}(A_{s+1},H_{s+1}),H_{s+1}\}
 \mid A_s=a_s,H_s=h_s
\right].
\]
We use the double index \((t+1,s)\) to distinguish the outcome time \(t+1\) from the recursion step \(s\). Although retaining both indices makes the notation somewhat more cumbersome, it allows the regression functions corresponding to different outcome times to be pooled and smoothed across \(t\) (see \S\ref{sec:estima-vd}). The target $\theta_{t+1}(\bar d_t)$ is identified as $\E[m_{t+1,1}\{d_1(A_1,H_1),H_1\}]$. However, $Y_{t+1}$ is unobserved when $R_{t+1}=0$, which implies that $m_{t+1,t}$ and hence all earlier regression functions are not identified by the observed data law without additional assumptions. To identify $\theta_{t+1}(\bar d_t)$, we introduce the following assumptions, which are extensions of the self-censoring model assumptions of \cite{chen2023causal}. The shadow variable restriction is depicted in Figure~\ref{fig:dag}.

\begin{assumption}[Shadow variable exclusion]\label{ass:shadow}
    Assume $A_t\indep R_{t+1}\mid H_t,C_{t+1}$.
\end{assumption}
\begin{assumption}[Positivity of missingness mechanism]\label{ass:positive}   
    Assume $\p(R_{t+1}=1\mid C_{t+1}, H_t)>0$ almost surely. 
\end{assumption}

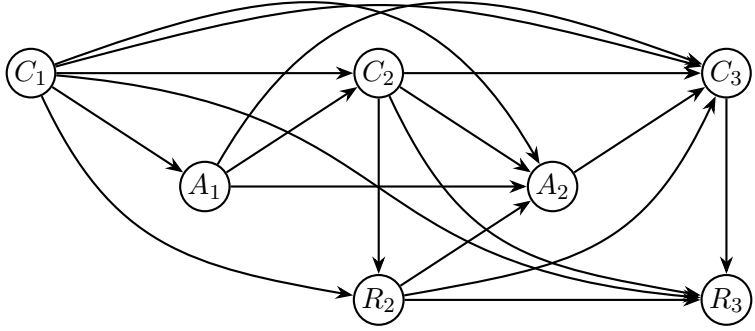
\begin{figure}[!htb]
\centering

\tikzset{
  >=Stealth,
  dag/.style={->, thick},
  var/.style={circle, draw, thick, minimum size=18pt, inner sep=1pt, align=center},
  treat/.style={rectangle, draw, thick, rounded corners=2pt,
                minimum height=18pt, minimum width=26pt, inner sep=2pt, align=center},
}

\begin{tikzpicture}[x=1cm,y=1cm]


\node[var]  (C1)  at (0,0)      {$C_1$};

\node[var](A1)  at (2.3,-1.5) {$A_1$};

\node[var]  (C2)  at (4.6,0)     {$C_2$};
\node[var]  (R2)  at (4.6,-3)  {$R_2$};
\node[var](A2)  at (6.9,-1.5)  {$A_2$};

\node[var]  (C3)   at (9.2,0)    {$C_3$};
\node[var]  (R3)  at (9.2,-3) {$R_3$};

\draw[dag] (C1) to (A1);
\draw[dag] (C1) to[out=20,in=120,looseness=1.25] (A2);
\draw[dag] (A1) to[out=60,in=160,looseness=1.25] (C3);
\draw[dag] (C2) to (A2);
\draw[dag] (R2) to (A2);
\draw[dag] (A2) to (C3);
\draw[dag] (C3) to (R3);
\draw[dag] (R2) to (R3);

\draw[dag] (C1) to (C2);
\draw[dag] (C1) to[out=-65,in=170,looseness=1.15] (R2);
\draw[dag] (C1) to[out=15,in=165,looseness=1.25] (C3);
\draw[dag] (C1) to[out=-5,in=175,looseness=1.35] (R3);

\draw[dag] (C2) to (C3);
\draw[dag] (C2) -- (R2);
\draw[dag] (C2) to[out=-65,in=170,looseness=1.20] (R3);
\draw[dag] (R2) to[out=10,in=-115,looseness=1.20] (C3);
\draw[dag] (A1) to (C2);
\draw[dag] (A1) to (A2);
\end{tikzpicture}%
\caption{Causal model illustration for visit-dependent treatments and $\tau = 2$.}\label{fig:dag}
\end{figure}

Assumption~\ref{ass:shadow} may be defensible in administrative health data when
\(C_{t+1}\) contains a sufficiently rich representation of the patient's health
status. Within strata of the observed history \(H_t\), the assumption states that
any association between the prior treatment vector \(A_t\) and the probability of
a subsequent visit \(R_{t+1}\) is mediated by the possibly unobserved covariate
vector \(C_{t+1}\), rather than arising through a direct pathway from treatment to
the visit process. The plausibility of the assumption can be strengthened by letting
\(C_{t+1}\) also capture relevant adverse effects, and other
post-treatment clinical features that influence the decision to seek or provide care. Likewise, the assumption may fail when treatments affect healthcare
utilization through mechanisms that are not adequately represented in
\(C_{t+1}\). For example, a treatment may cause addiction or adverse
effects that influence the likelihood of healthcare utilization. If those effects are not recorded in \(C_{t+1}\), the assumption would fail. The assumption also requires that
all common causes of \(A_t\) and \(R_{t+1}\) are captured in \(H_t\). Its
plausibility therefore depends not only on the richness of \(C_{t+1}\), but also
on whether \(H_t\) and \(C_{t+1}\) adequately represent the clinical, behavioral,
and administrative mechanisms jointly governing treatment and subsequent
healthcare encounters.

If these two assumptions hold, progress towards identification proceeds as follows. Define
\[
g_t(a_t\mid h_t) := \p(a_t\mid H_t=h_t).
\]
Letting \(g^d_t(\cdot\mid h_t)\) denote the conditional density or probability mass function of \(d_t(A_t, H_t)\mid H_t=h_t\), define
\[
\omega_t(a_t,h_t):=\frac{g^d_t(a_t\mid h_t)}{g_t(a_t\mid h_t)}.
\]
Note that $g_t^d$ need not exist, for example if $A_t$ is continuous and $d_t(a_t, h_t)$ is a constant function. Below we present an assumption that has been previously used to ensure existence of this density \cite{Haneuse2013}:
\begin{assumption}[Piecewise smooth invertibility]\label{ass:piecewise-invertibility}
For each \(t\) and each \(h_t\), the support of \(A_t\mid H_t=h_t\) can be
partitioned into subintervals $\mathcal I_{t,j}(h_t)$, $j=1,\ldots,J_t(h_t)$, such that, for \(a_t\in\mathcal I_{t,j}(h_t)\),
$
d_t(a_t,h_t)=d_{t,j}(a_t,h_t),
$
and each map \(d_{t,j}(\cdot,h_t)\) is one-to-one with inverse
\(b_{t,j}(\cdot,h_t)\) that is differentiable with respect to its first
argument.
\end{assumption}
Define the inverse-probability bridge function
\begin{equation}
\beta_t^0(H_t,C_{t+1})
:=
\frac{1}{\p(R_{t+1}=1\mid H_t,C_{t+1})}.\notag
\end{equation}
Furthermore, define the linear operator $T_t$ acting on functions $s(H_t, C_{t+1})$ by
\[
(T_t s)(a_t,h_t) := \E\{R_{t+1}s(H_t,C_{t+1})\mid A_t=a_t,H_t=h_t\}.
\]
The following theorem shows how this bridge can be used to express the unidentified regression $m_{t+1,t}$ as an observed-data pseudo-outcome regression.
\begin{theorem}[Latent bridge representation]\label{thm:latent_bridge}
Under Assumptions \ref{ass:shadow} and \ref{ass:positive}, the true inverse-probability bridge \(\beta_t^0\) satisfies
\begin{equation}\label{eq:fredholm_wor}
(T_t\beta_t^0)(A_t,H_t)=1\quad\text{a.s.}
\end{equation}
Moreover, the full-data regression obeys
\begin{equation}\label{eq:m_bridge_repr}
m_{t+1,t}(a_t,h_t)
=
\E[\eta_t(O)\mid A_t=a_t,H_t=h_t],
\end{equation}
where
\begin{equation}\label{eq:tildeY_def}
\eta_t(O) := R_{t+1}Y_{t+1}\beta_t^0(H_t,C_{t+1}).
\end{equation}
\end{theorem}

Note that, unlike propensity-based treatment bridge formulations in proximal inference \citep{ying2023proximal} and odds-based formulations in the shadow-variable and self-censoring literature \citep{miao2015identification,miao2016varieties,li2023self}, our bridge equation \eqref{eq:fredholm_wor} contains no auxiliary nuisance parameters. It can therefore be estimated directly as a single conditional moment restriction, without first estimating response or treatment probabilities and subsequently transforming them into potentially unstable odds or inverse probabilities.


Although the above result links the causal parameter to an observed data parameter, it is not yet a complete identification result as our assumptions do not guarantee uniqueness of the solution to the bridge equation \eqref{eq:fredholm_wor} and therefore different solutions can lead to different values of $m_{t+1, t}$. Instead of imposing conditions that guarantee uniqueness of $m_{t+1, t}$, we focus on weaker conditions that guarantee uniqueness of the target $\theta_{t+1}(\bar d_t)$. 

Let $T_t^\star$ denote the adjoint (e.g., transpose) of $T_t$ with respect to the inner products
\[
\langle u_1,u_2\rangle_{A,t}
:=
\E\!\left\{u_1(A_t,H_t)u_2(A_t,H_t)\right\},
\qquad
\langle s_1,s_2\rangle_{C,t}
:=
\E\!\left\{R_{t+1}\,s_1(H_t,C_{t+1})s_2(H_t,C_{t+1})\right\}.
\] That is,
\[
(T_t^\star u)(h_t,c_{t+1}) = \E\{u(A_t,H_t)\mid R_{t+1}=1,H_t=h_t,C_{t+1}=c_{t+1}\}.
\]
Define
\[
\Omega_s := \prod_{j=1}^s \omega_j(A_j,H_j).
\]
We will require the following assumption on the adjoint:
\begin{assumption}[Dual range condition for the representer]\label{ass:lambda_exist}
Let
\[
\phi_t(H_t,C_{t+1})
:=
\Omega_{t-1}\,
\E\!\left\{
Y_{t+1}\,\omega_t(A_t,H_t)
\mid R_{t+1}=1,H_t,C_{t+1}
\right\}.
\]
Assume that \(\phi_t\) belongs to the range of \(T_t^\star\), i.e.,
$
\phi_t \in \mathrm{Range}(T_t^\star).
$
\end{assumption}
Assumption~\ref{ass:lambda_exist} has been previously used in the proximal causal inference literature to guarantee $\sqrt{n}$-estimability \citep{zhang2023proximal}, and in self-censoring literature to derive identification conditions that do not require identification of the entire full data distribution \citep{li2023non}. Except in degenerate cases, this dual range condition implies the relevance condition $A_t\not\indep C_{t+1}\mid R_{t+1}=1, H_t$, which is often stated as an identification assumption \citep{miao2015identification}. However, the relevance condition is neither necessary nor sufficient for identification under our approach.
 
One can alternatively impose a completeness condition that guarantees uniqueness
of the solution to the bridge equation, but such a condition is strong and is not
required for our developments. The adjoint range condition used here does not imply uniqueness of the bridge function, but only
that the particular representer \(\phi_t\) belongs to the range of \(T_t^\star\).
This condition can nevertheless be restrictive. In particular, when the shadow
variable \(A_t\) contains little variation relative to the latent variable
\(C_{t+1}\), the range of \(T_t^\star\) may be too limited to contain
\(\phi_t\). Example~5 in Appendix~\ref{app:range-discrete} illustrates this issue
for a binary shadow variable, and the appendix provides additional
finite-dimensional examples concerning \Cref{ass:lambda_exist}.

The shadow-variable exclusion is not required  to obtain an inverse-probability representation such as \eqref{eq:fredholm_wor}. Indeed, one could redefine $\beta_t^0$ to also condition on $A_t$ in which case
the above theorem would hold without \Cref{ass:shadow}. 
 The shadow-variable restriction \Cref{ass:shadow} is instead closely connected to the plausibility of the above range condition. To see this, suppose that we dispensed with the shadow-variable assumption by allowing the bridge $\beta_t^0$ to depend on $A_t$. The corresponding bridge  equation
\[\E\!\left[\allowdisplaybreaks
R_{t+1}\beta_t(A_t,H_t,C_{t+1})
\mid A_t,H_t
\right]
=1\]
would hold without the shadow-variable assumption. However, the associated operator would map functions of $(A_t,H_t,C_{t+1})$ to functions of $(A_t,H_t)$. Its adjoint would therefore have range consisting only of functions that are measurable with respect to $(A_t,H_t)$. Consequently, the range condition would require the target loading $\phi_t(A_t,H_t,C_{t+1})$ to not depend on $C_{t+1}$, a restriction that will generally fail. The exclusion of $A_t$ from the bridge avoids this degeneracy. Because the same function of $(H_t,C_{t+1})$ must satisfy the bridge equation across different values of $A_t$, variation in the shadow variable provides multiple conditional moment restrictions on a common function. Accordingly, one way to make the range condition more plausible is to define \(A_t\) as a vector of shadow variables that includes additional treatment components, such
as dose or concomitant treatments, that are predictive of \(C_{t+1}\) among
complete cases. These components need not be of direct causal interest and may
be left unchanged by the modified treatment policy \(d_t\); their observed
variation can nevertheless enlarge the range of the adjoint operator and thereby
help satisfy the target-specific range condition. This strategy involves a
tradeoff, however, because every additional component included in \(A_t\) must
also satisfy the shadow-variable exclusion restriction for the visit process.
Thus, increasing the richness of \(A_t\) may strengthen the relevance of the
shadow variables while making the exclusion restriction more difficult to
justify.

Under Assumption~\ref{ass:lambda_exist}, there exists a square-integrable measurable function
\(\lambda_t(A_t,H_t)\) such that
\begin{equation}
    (T_t^\star \lambda_t)(H_t,C_{t+1})
=\phi_t(H_t, C_{t+1}), \label{eq:rep-iden}
\end{equation}
which is instrumental to prove the following theorem.
\begin{theorem}[Identification of the MTP outcome mean]\label{thm:theta_identified}
Under \Cref{ass:nuc,ass:support,ass:shadow,ass:positive,ass:lambda_exist}, let \(\beta_t(H_t,C_{t+1})\) be any square-integrable solution of \eqref{eq:fredholm_wor}. Then
\begin{align}\label{eq:theta_target_identified}
\theta_{t+1}(\bar d_t)
&=
\E\!\left\{
R_{t+1}\Omega_tY_{t+1}\beta_t(H_t,C_{t+1})
\right\}\\
&=
\E\!\left\{\lambda_t(A_t,H_t)\right\}.\notag
\end{align}
In particular, \(\theta_{t+1}(\bar d_t)\) is identified from the observed data law even if \(\beta_t\) and \(m_{t+1,t}(a_t,h_t)\) are not. Specifically, if \(\lambda_t\) and \(\widetilde\lambda_t\) are two square-integrable solutions to \eqref{eq:rep-iden}, then
\[
\E\{\lambda_t(A_t,H_t)\}
=
\E\{\widetilde\lambda_t(A_t,H_t)\}.
\]
Thus the right-hand side of \eqref{eq:theta_target_identified} does not depend on which representer is chosen and therefore $\theta_{t+1}(\bar d_t)$ is identified.
\end{theorem}
The above theorem provides identification using re-weighting of the outcome, analogous to an inverse-probability weighting identification result. The following corollary shows that the target may be also identified by taking any solution of the bridge equation,
forming the corresponding complete-case pseudo-outcome, and applying the usual
sequential-regression formula for the modified treatment policy.
\begin{corollary}[Sequential-regression identification]\label{cor:vd-sequential-regression}
Under \Cref{ass:nuc,ass:support,ass:shadow,ass:positive,ass:lambda_exist}, let
\(\beta_t(H_t,C_{t+1})\) be any square-integrable solution of the bridge equation
\eqref{eq:fredholm_wor}. Define
\[
\eta_t(O)
:=
R_{t+1}Y_{t+1}\beta_t(H_t,C_{t+1}).
\]
Set
\[
m_{t+1,t}(a_t,h_t)
:=
\E\{\eta_t(O)\mid A_t=a_t,H_t=h_t\},
\]
and, for \(s=t-1,\ldots,1\), define recursively
\[
m_{t+1,s}(a_s,h_s)
:=
\E\!\left[
 m_{t+1,s+1}\{d_{s+1}(A_{s+1},H_{s+1}),H_{s+1}\}
 \mid A_s=a_s,H_s=h_s
\right].
\]
Then, the target parameter is uniquely identified as
\begin{equation}
\theta_{t+1}(\bar d_t)
=
\E\!\left[
 m_{t+1,1}\{d_1(A_1,H_1),H_1\}
\right].\label{eq:iden2}    
\end{equation}

\end{corollary}

Having identified the target parameter, we next study its semiparametric
properties and develop estimators suitable for flexible, data-adaptive estimation of the nuisance functions. Following standard semiparametric theory
\citep{Bickel97}, we characterize the tangent space
induced by the bridge restriction and derive a gradient for the pathwise
derivative of \(\theta_{t+1}(\bar d_t)\). We then use this gradient to construct
cross-fitted influence-function-based estimators, including one-step and targeted
estimators, that permit the nuisance functions to be estimated using
machine-learning methods
\citep{pfanzagl1982contributions,van2011targeted,chernozhukov2018double}.
Our estimation strategy also builds on existing influence-function-based methods
for longitudinal causal parameters
\citep{luedtke2017sequential,rotnitzky2017multiply,diaz2023nonparametric}.
Finally, because the regression functions associated with different outcome
times share a common longitudinal structure, we employ the pool-and-smooth
strategy of \citet{susmann2025curve} to borrow information across \(t\) while
retaining target-specific estimation and inference.

\subsection{Semiparametric estimation theory}\label{sec:vd-semiparametric}

We next study the fixed-time parameter as a functional of the observed-data
law. Let \(\beta_t(H_t,C_{t+1})\) denote a selected square-integrable solution of \eqref{eq:fredholm_wor}. For $s<t$, let $q_{t+1,s}(a_s,h_s)
=
m_{t+1,s}\{d_s(a_s,h_s),h_s\}$ and 
\[
\eta_{t+1,s+1}(O)
:=q_{t+1,s+1}(A_{s+1},H_{s+1}).
\]
At $s=t$, let
\[
\eta_{t+1,t+1}(O)
:=R_{t+1}Y_{t+1}\beta_t(H_t,C_{t+1}).
\]
Thus the regressions in Corollary~\ref{cor:vd-sequential-regression} may be
written uniformly as
\[
m_{t+1,s}(a_s,h_s)
=
\E\{\eta_{t+1,s+1}(O)\mid A_s=a_s,H_s=h_s\}.
\]
Let $\lambda_t(A_t,H_t)$ be any square-integrable solution of
\eqref{eq:rep-iden}. In the full-data setting where $\eta_{t+1, t+1}$ is observed, \cite{diaz2023nonparametric} shows  the efficient influence function is equal to
\begin{equation}\notag
D_{t}^{\mathrm{seq}}(O)
:=
q_{t+1,1}(A_1,H_1)-\theta_{t+1}(\bar d_t)
+
\sum_{s=1}^t
\Omega_s\{\eta_{t+1,s+1}(O)-m_{t+1,s}(A_s,H_s)\}.
\end{equation}
Define the following function, which will be the contribution to the gradient from estimating the bridge function:
\begin{equation}\notag
D_t^{\mathrm{br}}(O)
:=
-\lambda_t(A_t,H_t)
\{R_{t+1}\beta_t(H_t,C_{t+1})-1\}.
\end{equation}
Our next objective is to construct estimating equations that permit flexible
estimation of the nuisance functions while supporting valid inference
for the target parameter. The following theorem derives a gradient of
the parameter, and the subsequent remainder expansion establishes
first-order insensitivity to nuisance-estimation error. Together, these
results provide the basis for cross-fitted estimation under suitable
product-rate conditions used in the construction of one-step estimators \citep{pfanzagl1982contributions,Bickel97}, targeted minimum loss-based estimators \citep{van2011targeted,van2014targeted}, and double/debiased machine learning \citep{chernozhukov2018double}.  These approaches permit flexible machine-learning estimation of the nuisance functions while retaining asymptotically linear, root-$n$-consistent inference, provided the nuisance estimators satisfy suitable product-rate conditions. The reader interested in the general relationship between gradients of
the pathwise derivative and Neyman orthogonality is referred to
\citet{chen2026equivalence}.
\begin{theorem}[Pathwise differentiability]
\label{thm:vd-if}
Under \Cref{ass:nuc,ass:support,ass:shadow,ass:positive,ass:lambda_exist},
as well as the regularity condition stated in
Appendix~\ref{proof:thm:vd-if}, the  parameter in \ref{eq:theta_target_identified}
is pathwise differentiable, with gradient given by $D_t(O)
=
D_t^{\mathrm{seq}}(O)+D_t^{\mathrm{br}}(O)$.
\end{theorem}


Theorem~\ref{thm:vd-if} establishes pathwise differentiability and provides
a gradient for estimation. Characterizing the full tangent space and
deriving the efficient influence function require accounting for the
restrictions imposed by the shadow-variable assumption and the adjoint
range condition. We base estimation directly on the derived gradient, which supports asymptotically linear inference under conditions on the estimated nuisances, without asserting semiparametric
efficiency. The following proposition establishes the required von Mises
expansion, showing that the gradient captures the first-order bias of the
plug-in estimator and bounding the remainder by products of
nuisance-estimation errors.
\begin{proposition}[Second-order remainder]
\label{prop:vd-second-order}
Let $P$ and $F$ be two distributions in the visit-dependent shadow-variable model and
assume the relevant density ratios, bridge functions, adjoint representers,
and sequential regressions are uniformly bounded. Let $m^{\beta_F}_{t+1,s,P}$ denote the regression functions obtained in the recursion when the conditional expectations are taken under $P$ but the pseudo-outcome is evaluated at $\beta_F$. Then
\[
\theta_{t+1}(\bar d_t;F)-\theta_{t+1}(\bar d_t;P)
=
-\E_P\{D_t(O;F)\}+R_{t,2}(P,F),
\]
where
\begin{align*}
|R_{t,2}(P,F)|
&\lesssim
\sum_{s=1}^t
\|\omega_{s,P}-\omega_{s,F}\|_{2,P}
\|m^{\beta_F}_{t+1,s,P}-m^{\beta_F}_{t+1,s,F}\|_{2,P}\\
&\quad+
\|\lambda_{t,P}-\lambda_{t,F}\|_{2,P}
\|\beta_{t,P}-\beta_{t,F}\|_{2,P}.
\end{align*}
\end{proposition}
Proposition~\ref{prop:vd-second-order} decomposes the first-order error of the
plug-in estimator into sequential-regression and bridge components. Our general
estimation strategy begins with the plug-in estimator
\(\theta_{t+1}(\bar d_t;\widehat P)\). By
Proposition~\ref{prop:vd-second-order}, its error admits the expansion
\[
\theta_{t+1}(\bar d_t;\widehat P)
-
\theta_{t+1}(\bar d_t;P)
=
-\E_P\!\left\{D_t(O;\widehat P)\right\}
+
R_{t,2}(P,\widehat P).
\]
Thus, up to the second-order remainder, the estimation error is given by minus
the population mean of the gradient evaluated at the estimated nuisance
functions. The estimators developed below address  how they estimate and
remove this first-order term. In particular, they separately correct the
sequential-regression and bridge components of the gradient.

\subsection{Estimators}\label{sec:estima-vd}

We address the problem of estimating the entire counterfactual mean curve $
\boldsymbol\theta
=
\{\theta_{t+1}(\bar d_t):t=1,\ldots,\tau\}
$ using the pool-and-smooth construction of \citet{susmann2025curve}. Here we outline the main idea, but we refer the interested reader to the original article for additional details. For each
outcome time \(t+1\), the identifying recursion involves the regressions $
m_{t+1,s}(a_s,h_s),
\qquad s=1,\ldots,t$. 
Thus, across all outcome times, these regressions form the lower-triangular
array $
\mathcal M
=
\{m_{t+1,s}:1\leq s\leq t\leq\tau\}$. 
We pool regressions that occur at the same relative position in their
respective backward recursions, i.e., at each diagonal of the array. Specifically, for
\(\ell=0,\ldots,\tau-1\), define the \(\ell\)th diagonal by
\[
\mathcal D_\ell
=
\{(t,s):t=\ell+1,\ldots,\tau,\ s=t-\ell\}.
\]
Here, \(\ell=t-s\) is the number of recursion steps between the current
regression time \(s\) and the terminal regression time \(t\). For example,
\(\mathcal D_0\) contains the terminal regressions $
m_{2,1},m_{3,2},\ldots,m_{\tau+1,\tau}$,
whereas \(\mathcal D_1\) contains $
m_{3,1},m_{4,2},\ldots,m_{\tau+1,\tau-1}$. More generally, \(\mathcal D_\ell\) contains 
\[ m_{\ell+2,1}(a_1, h_1),\,
m_{\ell+3,2}(a_2, h_2),\,
\ldots,\,
m_{\tau+1,\tau-\ell}(a_{\tau-\ell}, h_{\tau-\ell}).\]
Alternatively, this sequence may be written as 
$m_\ell(a,h,t) = m_{t+\ell+1, t}(a_t, h_t)$,
where $t\in\{1,\ldots,\tau-\ell\}$.
This re-indexing shows that the regressions on a given diagonal
\(\mathcal D_\ell\) can be estimated jointly by introducing a predictor $t$. The recursion is then defined iterating the index $\ell$. Define the corresponding
re-indexed pseudo-outcome by
$\eta_\ell(O,t)
:=
\eta_{t+\ell+1,t+1}(O)$, 
so that
\[
m_\ell(a,h,t)
=
\E\{\eta_\ell(O,t)\mid A_t=a,H_t=h\}.
\]
More explicitly,
\[
\eta_\ell(O,t)
=
\begin{cases}
R_{t+1}Y_{t+1}\beta_t(H_t,C_{t+1}),
& \ell=0,\\[4pt]
q_{t+\ell+1,t+1}(A_{t+1},H_{t+1}),
& \ell=1,\ldots,\tau-1.
\end{cases}
\]
The diagonals are estimated sequentially beginning with \(\ell=0\).
For \(\ell=0\), the pseudo-outcomes are the bridge-adjusted outcomes
at each outcome time. For \(\ell\geq1\), the pseudo-outcomes are formed
from the intervention-specific regressions estimated on the preceding
diagonal \(\mathcal D_{\ell-1}\).

A learner saturated in \(t\) recovers separate fitting as a special case, but including \(t\) as a predictor permits the regression surface to vary over
calendar time, while pooling data across $t$ allows us to borrow
information across outcome times. This can be particularly useful if the outcome is sparsely measured at any given time $t$, but there are multiple such times so that information can be borrowed.  

Proving asymptotic normality of the estimators while using general data adaptive estimators for the nuisance functions will generally require to estimate the nuisance functions using cross-fitting \citep{kennedy2024semiparametric}. To setup up cross-fitting, let \(\mathcal V_1,\ldots,\mathcal V_J\) denote a random partition of
\(\{1,\ldots,n\}\) into \(J\) validation samples, and let
\(\mathcal T_j=\{1,\ldots,n\}\setminus\mathcal V_j\) denote the corresponding
training sample. Let \(j(i)\) denote the training set corresponding to the validation fold containing observation
\(i\). For each fold \(j\), use \(\mathcal T_j\) to estimate the nuisance functions. For any nuisance function, we use the index $j$ to denote estimation using the corresponding training set. For instance, \(\widehat D_{t,j}^{\mathrm{br}}\) denotes the bridge component of the gradient evaluated using nuisance estimates that were trained in \(\mathcal T_j\).

Let \(\widehat\theta_{t+1}^{\mathrm{seq}}\) denote the cross-fitted
longitudinal TMLE or sequentially doubly robust estimator for the component
\(\theta_{t+1}(\bar d_t)\) discussed in \citet{susmann2025curve}, using these pooled-and-smoothed preliminary
regressions and treating the fitted bridge function as fixed
within each training fold. Our proposed estimator is defined as
\begin{equation}\label{eq:vd-one-step}
\widehat\theta_{t+1}
=
\widehat\theta_{t+1}^{\mathrm{seq}}
+
\frac{1}{n}\sum_{i=1}^n
\widehat D_{t,j(i)}^{\mathrm{br}}(O_i),
\qquad t=1,\ldots,\tau.
\end{equation}
The first term is constructed so that it removes the empirical first-order contribution of the nested
sequential regressions while borrowing information across outcome times in
their estimation. The second term corrects the remaining first-order bias from
estimating the bridge nuisance function. Under
cross-fitting, the usual remainder conditions for the pooled sequential
estimator and the product-rate condition
\(R_{t,2}(P,\widehat P)=o_p(n^{-1/2})\) for each \(t\), standard arguments \citep[e.g.,][]{kennedy2024semiparametric} show that
\[
\widehat\theta_{t+1}
-
\theta_{t+1}(\bar d_t)
=
\frac{1}{n}\sum_{i=1}^n D_t(O_i)
+
o_p(n^{-1/2}).
\]
Consequently,
\[
\sqrt n(\widehat{\boldsymbol\theta}
-\boldsymbol\theta)
\rightsquigarrow
N(0,\Sigma),
\qquad
\Sigma_{tr}=\operatorname{cov}\{D_t(O),D_r(O)\}.
\]
Pointwise standard errors are obtained from the empirical variances of the
cross-fitted gradients, and simultaneous confidence bands may be constructed
from their joint multiplier-bootstrap distribution as in
\citet{susmann2025curve}. Nuisance estimation for
\eqref{eq:vd-one-step}, including the selection of a canonical bridge when the
bridge equation admits multiple solutions, is discussed in
Section~\ref{sec:nuisance-estimation}.

%% file: full_main.tex
\section{Causal effects under full-history adjustment}\label{sec:full-history}

The developments in the previous section assume that the observed history is sufficient for confounding adjustment. This assumption is natural when treatment decisions are determined by information recorded in the healthcare system, but it may be inappropriate when latent values of time-varying covariates affect both treatment decisions and subsequent outcomes. For example, a choice between gastric bypass surgery and GLP-1 medication may be informed by a patient's evolving clinical state not recorded at missed encounters. In such settings, adequate confounding control may require adjustment for the full latent covariate history rather than only the recorded history. 

Let $F_s=(C_1,A_1,\ldots,C_{s-1},A_{s-1},C_s)$ denote the full-data history immediately preceding \((R_s,A_s)\), and let $J_s=\prod_{j=1}^sR_j$ with $J_0=1$ denote the indicator that every covariate value entering \(F_s\) has been observed. For an arbitrary function $d_s(a_s,f_s)$ depending on the current treatment and the \(k\)-lag full-data history, let \(\epsilon_s\sim Q_s\) be intervention randomizers that are independent across time and independent of all exogenous variables, with known laws \(Q_s\). We study a modified treatment policy \(d_s(a_s,r_s,f_s,\epsilon_s)\) defined as
\[
d_s(a_s,r_s,f_s,\epsilon_s)=
\begin{cases}
    d_s(a_s,f_s)\text{ if }r_s=1,\\
    t_s(a_s,\epsilon_s)\text{ if }r_s=0,
\end{cases}
\]
where \(t_s\) is a known treatment rule depending on a randomizer $\epsilon$. When $R_s=0$, $A_s=\alpha_s(H_s)$ by assumption, so that \(t_s\) defines how treatment will be continued at no-visit times. We have used $d$ to denote both functions in a slight abuse of notation, but their use will be clear from context. We use a stochastic continuation rule in this section because identification requires evaluation of the treatment rule separately for the \(R_s=0\) and \(R_s=1\) cases. If \(A_s\) is continuous and treatment at the treatment policy a no-visit time is deterministically equal to \(\alpha_s(H_s)\), the resulting identifying functional would be a point evaluation, and therefore would not be pathwise differentiable. We discuss this further in \S\ref{sec:fh-semiparametric}. Below we present an example where such stochastic continuation rule may be of scientific interest.

\begin{example}[GLP-1 continuation under a surgery-versus-medication strategy]
Let \(A_s=(B_s,G_s)\), where \(B_s\) indicates bariatric surgery status and \(G_s\) denotes GLP-1 prescription dose. At visit times, the policy \(d_s\) may assign surgery or modify treatment. At times without a visit, surgery status cannot change, but we may consider a regimen that allows GLP-1 dose to vary randomly around the continued prescription dose. Specifically, let \(\epsilon_s\sim Q_s\) be a positive randomizer with known distribution centered at one, for example uniform on \([1-\delta,1+\delta]\). Then, for \(R_s=0\), define
\[
t_s\{(B_s,G_s),\epsilon_s\}
=
(B_s,G_s\epsilon_s),
\]
where \(G_s=\alpha_s(H_s)\) denotes the treatment level implied by the patient's known prescription history, for example the dose and number of pills remaining based on prior prescriptions and other observed historical information.
\end{example}

As before, the causal target is the expectation of the outcome under this modified treatment policy, where the expectation also averages over the intervention randomizers:
\[
\theta_{t+1}(\bar d_t)
=
\E\{Y_{t+1}(\bar d_t)\}.
\]

\subsection{Identification assumptions}

We next discuss the assumptions needed to identify the causal functional above. In addition to the shadow-variable assumptions introduced in the previous section, the full-history setting requires stronger assumptions on the visit process. There are two related reasons. First, as in the observed-history setting, the modified treatment policy depends on the natural value of the visit indicator \(R_s\), with a distinct stochastic continuation rule when no visit occurs. Consequently, \(R_s\) must be sequentially adjusted for throughout the causal identification proof even when it may not otherwise be needed for confounding adjustment. Second, adequate confounding control is now assumed to require adjustment for the full-data history \(F_s\), which is not fully observed unless \(J_s=1\). Sequentially adjusting for \(\bar R_s\) while also conditioning on \(J_s=1\) cannot generally be justified without additional assumptions. Below, we impose two assumptions that allow us to do so.

First, we address confounding by assuming that \(F_s\) contains all common causes of the current visit--treatment pair \((R_s,A_s)\) and subsequent variables, so that the pair is conditionally exchangeable given \(F_s\). This is stronger than in the observed-history setting, where adjustment can be based directly on the observed history, including the observed visit pattern.

\begin{assumption}[Sequential randomization conditional on the full data]\label{ass:fh-nuc-markov}
For every \(s\leq t\), assume
\[
(U_{R,s}, U_{A,s})
\indep
\left(
\underbar U_{C,s+1},
\underbar U_{R,s+1},
\underbar U_{A,s+1}
\right)
\mid
F_s.
\]
\end{assumption}

The next assumption is needed because adjusting for $F_s$ requires conditioning on \(J_s=1\) (or, recursively, on \(R_s=1\)). Such a restriction would generally not identify the desired causal functional if \(R_s\) were itself a confounder of the effect of \(A_s\) on subsequent variables. We therefore impose an exclusion restriction requiring that $R_s$ can only have an effect on future variables through the current treatment vector $A_s$. Together with \Cref{ass:fh-nuc-markov}, this allows \(R_s\) to determine the treatment modification through \(d_s\), while ensuring that the subsequent causal recursion can be evaluated on the \(RJs=1\) subset.

\begin{assumption}[Exclusion restriction on visit indicator]\label{ass:fh-measurement-exchange}
For every $s\leq t+1$, assume the structural functions $\{f_{C,s},f_{R,s},f_{A,s}\}$ do not depend on $\bar R_{s-1}$.
\end{assumption}
A plausible justification for this assumption is that $A_s$ is intended to capture all treatments and procedures that affect a patient's underlying health status, so any effect of healthcare utilization at time $s$ is explained by it with no additional direct effect of engaging with the healthcare system. In particular, $A_s$ should contain healthcare-process and access-related information that governs subsequent encounters and treatment decisions. For instance, prior visits may affect the probability of a future visit through mechanisms such as scheduled return intervals, referrals, pending laboratory or imaging orders, or other follow-up instructions. These variables should be included in $A_s$. Likewise, persistent determinants of healthcare utilization, such as insurance coverage, socioeconomic and other social determinants of health, transportation or geographic access, and healthcare-seeking behavior, may be included when available.

As in \Cref{ass:visit-dep}, treatment status must also be known when no visit occurs. Unlike \Cref{ass:visit-dep}, we let the treatment value at a visit time be affected by unmeasured variables. Also, unlike the previous section, this assumption will play a crucial role in this section because adjustment for $\bar R_s$ is not possible, and therefore we need to separate the branches of the outcome regressions that depend on $R_s$ through $d_s$ from those that do not.

\begin{assumption}[Known treatment at no-visit times]\label{ass:fh-known-novisit}
Assume that the structural equation for $A_t$ satisfies the restriction
\[
    A_t =
    \begin{cases}
        f_{A_t}^{(1)}(\bar A_{t-1}, \bar R_{t-1}, \bar C_t, U_{A,t}),
        & R_t=1,\\
        \alpha_t(H_t), & R_t=0,
    \end{cases}
\]
where \(\alpha_t\) is a known deterministic continuation rule that determines treatment status at a no-visit time \(t\) from the observed history. 
\end{assumption}

We then require the standard positivity assumption required to apply \Cref{ass:fh-nuc-markov}:

\begin{assumption}[Supported intervention on the full data]\label{ass:fh-support-markov}
If $(a_1,c_1)\in\supp(A_1,C_1)$, then $(d_1(a_1,1,c_1),c_1)\in\supp(A_1,C_1)$. For each $s=2,\ldots,t$, if
\[
(r_s,a_s,c_s)\in \supp\{R_s,A_s,C_s\mid A_{s-1}=a_{s-1}^d,J_{s-1}=1,F_{s-1}=f^d_{s-1}\},
\]
then, with $a_s^d=d_s\{a_s,r_s,f_s^d,\epsilon_s\}$ and $f_s^d=(f_{s-1}^d,a_{s-1}^d,c_s)$,
\[
(1,a_s^d,c_s)\in \supp\{R_s,A_s,C_s\mid A_{s-1}=a_{s-1}^d,J_{s-1}=1,F_{s-1}=f^d_{s-1}\}
\]
for \(Q_s\)-almost every \(\epsilon_s\). In addition, the stochastic no-visit treatment law induced by \(t_s(A_s,\epsilon_s)\) is absolutely continuous with respect to the corresponding treatment law among complete cases.
\end{assumption}
Finally, we impose the same shadow-variable structure used in the observed-history section, replacing the observed history by the complete full-data state. These assumptions are parallel to those in the first section. In both sections, \(A_s\) is the shadow variable and \(C_{s+1}\) is the latent argument of the bridge.  A directed acyclic graph depicting some of these assumptions is given in Figure~\ref{fig:dag2}.

\begin{assumption}[Full-history sequential shadow variable exclusion]\label{ass:fh-shadow-markov}
For every \(s\leq t\), assume
\[
A_s
\indep
R_{s+1}
\mid
J_s=1,F_s,C_{s+1}.
\]
\end{assumption}

\begin{assumption}[Full-history visit positivity]\label{ass:fh-positive-markov}
For every \(s\leq t\),
\[
\P\{R_{s+1}=1
\mid
J_s=1,F_s,C_{s+1}\}>0
\]
almost surely on the relevant support.
\end{assumption}

The developments below will require estimation of regression functions that condition on $F_s$ and therefore on the  event $J_s=1$. This event may become rare quickly. We therefore impose finite-memory restrictions that reduce the dimension of the conditioning set $F_s$, allowing us to condition on a better supported observation event. Fix \(k\geq1\), let $\ell_s=(s-k+1)\vee1$, and define $F_s^{(k)}=(C_{\ell_s},A_{\ell_s},\ldots,C_{s-1},A_{s-1},C_s)$, with $F_1^{(k)}=C_1$. Let $J_s^{(k)}=\prod_{j=\ell_s}^sR_j$ be the local complete-history indicator, and define $K_s^{(k)}=J_{s-1}^{(k)}R_s$. By convention, when \(k=\infty\), we recover the entire history so that \(F_s^{(k)}=F_s\) and \(J_s^{(k)}=K_s^{(k)}=J_s\). For compactness, write
\[
\mathcal R_s^{(k)}=(C_1,A_1,R_1,\ldots,C_{\ell_s-1},A_{\ell_s-1},R_{\ell_s-1})
\]
for the latent and visit history that precedes the current \(k\)-lag state, with the convention that this vector is empty when \(\ell_s=1\).

\begin{assumption}[Markov assumptions for full-history adjustment]\label{ass:fh-markov}
For each relevant \(s\), assume
\begin{align}
C_{s+1}
&\indep
\mathcal R_s^{(k)}
\mid
F_s^{(k)},A_s,
\tag{M1}\label{ass:fh-markov-M1}\\
A_s
&\indep
\mathcal R_s^{(k)}
\mid
J_s^{(k)}=1,F_s^{(k)},
\tag{M2}\label{ass:fh-markov-M2}\\
R_{s+1}
&\indep
\mathcal R_s^{(k)}
\mid
J_s^{(k)}=1,F_s^{(k)},A_s,C_{s+1}.
\tag{M3}\label{ass:fh-markov-M3}
\end{align}
\end{assumption}
For finite \(k\), we also require that the known value \(\alpha_s(H_s)\) needed to construct the stochastic no-visit continuation rule be determined by the retained state. If the continuation rule uses additional observed history, the retained state can be augmented with those observed variables without changing the arguments below.

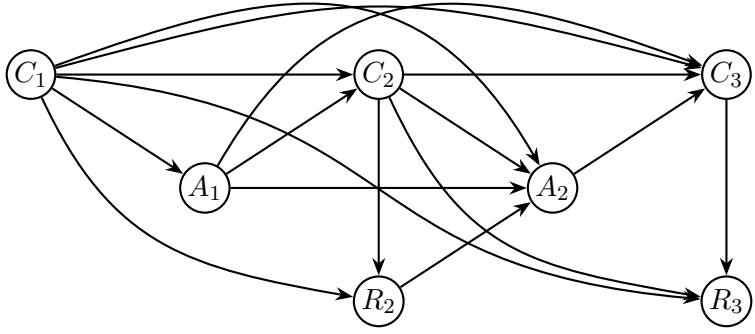
\begin{figure}[!htb]
\centering

\tikzset{
  >=Stealth,
  dag/.style={->, thick},
  var/.style={circle, draw, thick, minimum size=18pt, inner sep=1pt, align=center},
  treat/.style={rectangle, draw, thick, rounded corners=2pt,
                minimum height=18pt, minimum width=26pt, inner sep=2pt, align=center},
}

\begin{tikzpicture}[x=1cm,y=1cm]


\node[var]  (C1)  at (0,0)      {$C_1$};

\node[var](A1)  at (2.3,-1.5) {$A_1$};

\node[var]  (C2)  at (4.6,0)     {$C_2$};
\node[var]  (R2)  at (4.6,-3)  {$R_2$};
\node[var](A2)  at (6.9,-1.5)  {$A_2$};

\node[var]  (C3)   at (9.2,0)    {$C_3$};
\node[var]  (R3)  at (9.2,-3) {$R_3$};

\draw[dag] (C1) to (A1);
\draw[dag] (C1) to[out=20,in=120,looseness=1.25] (A2);
\draw[dag] (A1) to[out=60,in=160,looseness=1.25] (C3);
\draw[dag] (C2) to (A2);
\draw[dag] (R2) to (A2);
\draw[dag] (A2) to (C3);
\draw[dag] (C3) to (R3);

\draw[dag] (C1) to (C2);
\draw[dag] (C1) to[out=-65,in=170,looseness=1.15] (R2);
\draw[dag] (C1) to[out=15,in=165,looseness=1.25] (C3);
\draw[dag] (C1) to[out=-5,in=175,looseness=1.35] (R3);

\draw[dag] (C2) to (C3);
\draw[dag] (C2) -- (R2);
\draw[dag] (C2) to[out=-65,in=170,looseness=1.20] (R3);
\draw[dag] (A1) to (C2);
\draw[dag] (A1) to (A2);
\end{tikzpicture}%
\caption{Causal model illustration for visit-indepedent exposures and $\tau = 2$. Note the additional exclusion restrictions on the monitoring process.}\label{fig:dag2}
\end{figure}

\subsection{Identification}

For each \(s\leq t\), define the true inverse-probability bridge
\begin{equation}
\beta_s^0(F_s^{(k)},C_{s+1})
:=
\frac{1}{
\P\{R_{s+1}=1
\mid
J_s^{(k)}=1,F_s^{(k)},C_{s+1}\}}
\notag
\end{equation}
and the linear operator
\begin{equation}
(T_sh_s)(a_s,f_s)
:=
\E\!\left[
R_{s+1}h_s(F_s^{(k)},C_{s+1})
\mid
J_s^{(k)}=1,A_s=a_s,F_s^{(k)}=f_s
\right].
\notag
\end{equation}

\begin{lemma}[Bridge equation under full-history adjustment]\label{lem:fh-bridge-equation}
Under \Cref{ass:fh-shadow-markov,ass:fh-positive-markov}, $\beta_s^0$ is a solution of the following bridge equation:
\begin{equation}
(T_s\beta_s)(A_s,F_s^{(k)})=1.
\label{eq:fh-bridge-equation}
\end{equation}
\end{lemma}

As in the first section, the bridge equation need not have a unique solution. Let \(\beta_s(F_s^{(k)},C_{s+1})\) denote an arbitrary square-integrable solution of the bridge equation \eqref{eq:fh-bridge-equation}. Set
\begin{equation}
q_{t+1,t+1}(A_{t+1}, F_{t+1}^{(k)}):=R_{t+1}\beta_t(F_t^{(k)},C_{t+1})Y_{t+1}.
\label{eq:fh-terminal-q}
\end{equation}
Proceeding backward for \(s=t,\ldots,1\), define
\begin{equation}
m_{t+1,s}(a_s,f_s)
:=
\E\!\left[
q_{t+1,s+1}(A_{s+1}, F_{s+1}^{(k)},\epsilon_{s+1})
\mid
J_s^{(k)}=1,A_s=a_s,F_s^{(k)}=f_s
\right],
\label{eq:fh-m-recursion}
\end{equation}
where the \(\epsilon_{s+1}\) argument is omitted at the terminal step. Here and throughout, expectations of pseudo-outcomes involving an intervention randomizer also average over its known law \(Q_{s+1}\), which is independent of the observed-data law. For \(s\geq2\), define the pseudo-outcome
{\small
\begin{equation}
q_{t+1,s}(A_s, F_s^{(k)},\epsilon_s)
:=R_s\bigg[
m_{t+1,s}\{d_s(A_s,1,F_s^{(k)}),F_s^{(k)}\}+\{\beta_{s-1}(F_{s-1}^{(k)},C_s)-1\}
 m_{t+1,s}\{t_s(\alpha_s(H_s^{(k)}),\epsilon_s),F_s^{(k)}\}
\bigg].
\label{eq:fh-q-recursion}
\end{equation}}%
Define $g_s(a_s\mid f_s)$ as the density or probability mass function of $A_s$ conditional on $\{J_s^{(k)}=1,F_s^{(k)}=f_s\}$, and let $g_s^d(\cdot\mid f_s)$ denote the pushforward of $g_s(\cdot\mid f_s)$ through $d_s(\cdot,1,f_s)$. Let $g_s^t(\cdot\mid f_s)$ denote the density or probability mass function induced by \(t_s\{\alpha_s(H_s^{(k)}),\epsilon_s\}\), with \(\epsilon_s\sim Q_s\), conditional on \(F_s^{(k)}=f_s\). Define
\begin{equation}
\omega_s^d(a_s,f_s)
:=
\frac{g_s^d(a_s\mid f_s)}{g_s(a_s\mid f_s)},\qquad
\omega_s^t(a_s,f_s)
:=
\frac{g_s^t(a_s\mid f_s)}{g_s(a_s\mid f_s)}. \label{eq:fh-omega-t} 
\end{equation}
Set $\Omega_1:=\omega_1^d(A_1,C_1)$, and, for $s\geq2$, define recursively
\begin{equation}
\Omega_s
:=
\Omega_{s-1}R_s\left[
\omega_s^d(A_s,F_s^{(k)})
+\{\beta_{s-1}(F_{s-1}^{(k)},C_s)-1\}
\omega_s^t(A_s,F_s^{(k)})
\right].
\label{eq:fh-Omega}
\end{equation}
Let \(T_s^\star\) denote the adjoint of \(T_s\) with respect to the inner products
\begin{align*}
\langle u_1,u_2\rangle_{A,s}
&:=
\E\!\left[J_s^{(k)}u_1(A_s,F_s^{(k)})u_2(A_s,F_s^{(k)})\right],\\
\langle h_1,h_2\rangle_{C,s}
&:=
\E\!\left[K_{s+1}^{(k)}h_1(F_s^{(k)},C_{s+1})h_2(F_s^{(k)},C_{s+1})\right].
\end{align*}
Thus
\begin{equation}
(T_s^\star u)(f_s,c_{s+1})
=
\E\!\left[
 u(A_s,F_s^{(k)})
\mid
K_{s+1}^{(k)}=1,F_s^{(k)}=f_s,C_{s+1}=c_{s+1}
\right].
\label{eq:fh-adjoint}
\end{equation}

As in the observed-history section, we impose a target-specific range condition rather than requiring the bridge to be unique. The appendix shows that replacing one bridge solution while keeping the remaining bridge functions fixed changes the baseline sequential-regression expression through a single target-specific loading.

For every \(s\leq t\), define
\begin{equation}
\nu_{t+1,s+1}
:=
\begin{cases}
m_{t+1,s+1}\{t_{s+1}(\alpha_{s+1}(H_{s+1}^{(k)}),\epsilon_{s+1}),F_{s+1}^{(k)}\}, & s<t,\\
Y_{t+1}, & s=t.
\end{cases}
\label{eq:fh-nu}
\end{equation}

\begin{assumption}[Sequential adjoint range condition under full-history adjustment]\label{ass:fh-target-dual-range}
For every square-integrable bridge collection satisfying \eqref{eq:fh-bridge-equation} and every \(s\leq t\), define
\begin{equation}
\phi_{t+1,s}(F_s^{(k)},C_{s+1})
:=
\E\!\left[
\Omega_s\nu_{t+1,s+1}
\mid
K_{s+1}^{(k)}=1,F_s^{(k)},C_{s+1}
\right].
\label{eq:fh-phi-loading}
\end{equation}
Assume that \(\phi_{t+1,s}\) is square-integrable and
$\phi_{t+1,s}\in\operatorname{Range}(T_s^\star)$.
\end{assumption}

This condition has the same role as \Cref{ass:lambda_exist}. If two functions solve the same bridge equation, their difference belongs to \(\operatorname{Null}(T_s)\); the adjoint range condition ensures that the target-specific loading is orthogonal to that unidentified direction.

\begin{theorem}[Identification under full-history adjustment]\label{thm:fh-target-identification}
Under \Cref{ass:fh-nuc-markov,ass:fh-measurement-exchange,ass:fh-known-novisit,ass:fh-support-markov,ass:fh-shadow-markov,ass:fh-positive-markov,ass:fh-markov,ass:fh-target-dual-range}, let \(\beta_s\) be any square-integrable solution of \eqref{eq:fh-bridge-equation} for every \(s\leq t\), and define \(q_{t+1,s}\) and \(m_{t+1,s}\) by \eqref{eq:fh-terminal-q}--\eqref{eq:fh-q-recursion}. Then
\begin{equation}
\theta_{t+1}(\bar d_t)
=
\E\!\left[
 m_{t+1,1}\{d_1(A_1,1,C_1),C_1\}
\right].
\label{eq:fh-theta}
\end{equation}
The right-hand side is invariant to the sequence of bridge solutions selected. Thus neither the individual bridge functions nor the intermediate regression functions need be uniquely identified.
\end{theorem}

\subsection{Semiparametric estimation theory}\label{sec:fh-semiparametric}

We next study the fixed-time parameter as a functional of the observed-data law. Let \(\beta_s(F_s^{(k)},C_{s+1})\), \(s=1,\ldots,t\), denote selected square-integrable solutions of \eqref{eq:fh-bridge-equation}. For \(s=1,\ldots,t\), define
\[
\eta_{t+1,s+1}(O)
:=
q_{t+1,s+1}(A_{s+1},F_{s+1}^{(k)},\epsilon_{s+1}),
\]
where \(A_{t+1}\) and \(\epsilon_{t+1}\) are omitted at the terminal step. Thus \eqref{eq:fh-m-recursion} may be written uniformly as
\[
m_{t+1,s}(a_s,f_s)
=
\E\!\left[
\eta_{t+1,s+1}(O)
\mid
J_s^{(k)}=1,A_s=a_s,F_s^{(k)}=f_s
\right].
\]
Holding the selected bridge functions fixed, define
\begin{equation}
\begin{split}
D_t^{\mathrm{seq}}(O)
:={}&
m_{t+1,1}\{d_1(A_1,1,C_1),C_1\}
-\theta_{t+1}(\bar d_t)\\
&+
\sum_{s=1}^t
\Omega_s
\left\{
\eta_{t+1,s+1}(O)
-
m_{t+1,s}(A_s,F_s^{(k)})
\right\}.
\end{split}
\label{eq:fh-seq-gradient}
\end{equation}
Note that if \(A_s\) is continuous and \(t_s\) assigns treatment deterministically equal to \(\alpha_s(H_s)\), then \(g_s^t\) is a point mass and is generally not absolutely continuous with respect to \(g_s\). Consequently, \(\omega_s^t\), and hence \(\Omega_s\) in \eqref{eq:fh-seq-gradient}, are not well defined. More fundamentally, the resulting point-evaluation functional is generally not pathwise differentiable under the nonparametric model considered here, so a regular gradient does not exist. This motivates our introduction of the randomizer in the continuation rule \(t_s\).

For every \(s\leq t\), let \(\lambda_{t+1,s}(A_s,F_s^{(k)})\) be a square-integrable solution of $
T_s^\star\lambda_{t+1,s}
=
\phi_{t+1,s}$, whose existence follows from \Cref{ass:fh-target-dual-range}, and define
\begin{equation}
D_t^{\mathrm{br}}(O)
:=
-\sum_{s=1}^t
J_s^{(k)}
\lambda_{t+1,s}(A_s,F_s^{(k)})
\left\{
R_{s+1}\beta_s(F_s^{(k)},C_{s+1})-1
\right\}.
\label{eq:fh-bridge-gradient}
\end{equation}

The following theorem presents a gradient of the pathwise derivative that will be used for estimation.
\begin{theorem}[Pathwise differentiability]
\label{thm:fh-if}
Under \Cref{ass:fh-nuc-markov,ass:fh-measurement-exchange,ass:fh-known-novisit,ass:fh-support-markov,ass:fh-shadow-markov,ass:fh-positive-markov,ass:fh-markov,ass:fh-target-dual-range},
as well as the regularity condition stated in
Appendix~\ref{proof:thm:fh-if}, the  parameter in \cref{eq:iden2} is
pathwise differentiable, with gradient given by $D_t(O)
=
D_t^{\mathrm{seq}}(O)+D_t^{\mathrm{br}}(O)$.
\end{theorem}

As in the observed-history setting, we base estimation directly on the derived
gradient rather than on its generally implicit projection onto the tangent
space. The following proposition gives the corresponding second-order
expansion.

\begin{proposition}[Full-history second-order remainder]
\label{prop:fh-second-order}
Let \(P\) and \(F\) be two distributions in the full-history shadow-variable
model and assume the relevant density ratios, bridge functions, adjoint
representers, and sequential regressions are uniformly bounded. For
\(X\in\{P,F\}\), write
\[
\bar\beta_X=(\beta_{1,X},\ldots,\beta_{t,X}),
\]
and let \(m^{\bar\beta_F}_{t+1,s,P}\) denote the regression functions obtained
in the recursion when the conditional expectations are taken under \(P\) but
the pseudo-outcomes are evaluated using the bridge collection
\(\bar\beta_F\). Then
\[
\theta_{t+1}(\bar d_t;F)-\theta_{t+1}(\bar d_t;P)
=
-\E_P\{D_t(O;F)\}+R_{t,2}(P,F),
\]
where
\begin{align*}
|R_{t,2}(P,F)|
&\lesssim
\sum_{s=1}^t
\|\omega^d_{s,P}-\omega^d_{s,F}\|_{2,P}
\|m^{\bar\beta_F}_{t+1,s,P}-m^{\bar\beta_F}_{t+1,s,F}\|_{2,P}\\
&\quad+
\sum_{s=2}^t
\|\omega^t_{s,P}-\omega^t_{s,F}\|_{2,P}
\|m^{\bar\beta_F}_{t+1,s,P}-m^{\bar\beta_F}_{t+1,s,F}\|_{2,P}\\
&\quad+
\sum_{s=1}^t
\|\lambda_{t+1,s,P}-\lambda_{t+1,s,F}\|_{2,P}
\|\beta_{s,P}-\beta_{s,F}\|_{2,P}\\
&\quad+
\sum_{1\leq r<s\leq t}
\|\beta_{r,P}-\beta_{r,F}\|_{2,P}
\|\beta_{s,P}-\beta_{s,F}\|_{2,P}.
\end{align*}
\end{proposition}

The first two terms are the analogues of the sequential-regression remainder
in Proposition~\ref{prop:vd-second-order}, with separate errors for the
visit-time and no-visit policy density ratios. The third term is the analogue
of the bridge--adjoint product in the observed-history result. The final term
is specific to the full-history recursion: because a different bridge enters
at every time point, changing several bridge functions simultaneously produces
mixed products of their estimation errors.

Thus, for a plug-in estimator based on \(\widehat P\),
\[
\theta_{t+1}(\bar d_t;\widehat P)
-
\theta_{t+1}(\bar d_t;P)
=
-\E_P\!\left\{D_t(O;\widehat P)\right\}
+
R_{t,2}(P,\widehat P).
\]
As before, the estimators below remove the empirical first-order contribution
of the gradient, while the proposition gives conditions under which the
remaining bias is \(o_p(n^{-1/2})\).

\subsection{Estimators}\label{sec:estima-fh}

We use the same pool-and-smooth and cross-fitting construction introduced in
\S\ref{sec:estima-vd}. In particular, the lower-triangular regression array,
its diagonals \(\mathcal D_\ell\), and the re-indexing by \(\ell=t-s\) are
defined exactly as in that section. The difference is that the full-history
regressions condition on the complete retained state. Under the same
re-indexing,
\begin{equation}
m_\ell(a,f,t)
=
\E\!\left[
\eta_\ell(O,t)
\mid
J_t^{(k)}=1,A_t=a,F_t^{(k)}=f
\right],
\label{eq:fh-pooled-regression}
\end{equation}
where
\begin{equation}
\eta_\ell(O,t)
=
\begin{cases}
R_{t+1}\beta_t(F_t^{(k)},C_{t+1})Y_{t+1},
& \ell=0,\\[2mm]
q_{t+\ell+1,t+1}(A_{t+1},F_{t+1}^{(k)},\epsilon_{t+1}),
& \ell=1,\ldots,\tau-1.
\end{cases}
\notag
\end{equation}
For \(\ell\geq1\), the second line can be written using the regression on the
preceding diagonal as
\begin{align}
\eta_\ell(O,t)
=
R_{t+1}\bigg[
&
m_{\ell-1}\{d_{t+1}(A_{t+1},1,F_{t+1}^{(k)}),F_{t+1}^{(k)},t+1\}
\nonumber\\[-1mm]
&+
\{\beta_t(F_t^{(k)},C_{t+1})-1\}
m_{\ell-1}\{t_{t+1}(\alpha_{t+1}(H_{t+1}^{(k)}),\epsilon_{t+1}),
F_{t+1}^{(k)},t+1\}
\bigg].
\notag
\end{align}
As throughout this section, the regression in
\eqref{eq:fh-pooled-regression} averages over the known law of the
intervention randomizer. The diagonals are estimated sequentially in the same
way as in \S\ref{sec:estima-vd}.

Using the same folds \(\mathcal V_j,\mathcal T_j\), and notation \(j(i)\)
defined in \S\ref{sec:estima-vd}, each training sample is used to estimate the
bridge functions \(\beta_s\), the treatment laws entering
\(\omega_s^d\) and \(\omega_s^t\), the pooled regressions in
\eqref{eq:fh-pooled-regression}, and the adjoint functions
\(\lambda_{t+1,s}\). The loadings \(\phi_{t+1,s}\) are estimated using
\eqref{eq:fh-phi-loading}; these estimates determine the propagation factors
\(\widehat\Omega_s\) and the fold-specific gradient components.

Let \(\widehat\theta_{t+1}^{\mathrm{seq}}\) denote the cross-fitted
longitudinal TMLE or sequentially doubly robust estimator constructed as in
\S\ref{sec:estima-vd}, treating the fitted bridge collection as fixed within
each training fold. The full-history estimator is
\begin{equation}
\widehat\theta_{t+1}
=
\widehat\theta_{t+1}^{\mathrm{seq}}
+
\frac1n
\sum_{i=1}^n
\widehat D_{t,j(i)}^{\mathrm{br}}(O_i),
\qquad
t=1,\ldots,\tau.
\label{eq:fh-one-step}
\end{equation}
Under cross-fitting, consistency of the estimated gradient, and product-rate
conditions that make the remainder in Proposition~\ref{prop:fh-second-order}
\(o_p(n^{-1/2})\), the same asymptotic linearity, covariance estimation, and
simultaneous confidence-band construction described in
\S\ref{sec:estima-vd} apply.

%% file: estimation_main.tex
\section{Nuisance-function estimation}\label{sec:nuisance-estimation}

The observed-history and full-history analyses use different identifying
recursions and different gradients, but their nuisance functions can be
estimated using a common collection of tools. In both settings, estimation
requires conditional regressions, policy density ratios, solutions of bridge
conditional moment equations, and solutions of the corresponding adjoint
conditional moment equations. This section describes practical approaches to
these tasks. The construction of the final estimators, including the
pool-and-smooth organization of the sequential regressions, is given in
Sections~\ref{sec:estima-vd} and \ref{sec:estima-fh}.

\subsection{Conditional regressions}

All conditional expectations entering the identifying recursions may be
estimated with standard regression methods, preferably with cross-fitting. The full-history
regressions are fit among observations with a complete retained state
\(J_s^{(k)}=1\). Whenever a pseudo-outcome contains an intervention randomizer,
the regression target also averages over its known law. In practice, this
expectation may be evaluated analytically when available or approximated by
Monte Carlo draws of the randomizer within the training fold.

For estimation of the entire curve, these regression problems are not fit
separately for every pair \((t,s)\). Instead, Sections~\ref{sec:estima-vd} and
\ref{sec:estima-fh} organize them by the lag \(\ell=t-s\), stack regression
problems on a common diagonal, and include calendar time as a predictor. A
learner saturated in calendar time recovers separate fitting as a special
case, whereas smoother learners can borrow information across outcome times.

\subsection{Policy density ratios}

\subsubsection{Observed-history adjustment}

In the observed-history setting, the policy density ratio at time \(s\) is
\[
\omega_s(a_s,h_s)
=
\frac{g_s^d(a_s\mid h_s)}{g_s(a_s\mid h_s)}.
\]
A convenient estimator is based on probabilistic classification
\citep{qin1998inferences,diaz2023nonparametric}. For each observation, construct
a numerator row $\{A_{s,i}'=d_s(A_{s,i},H_{s,i}),H_{s,i},\Delta_{s,i}=1\}$
and a denominator row $
\{A_{s,i}'=A_{s,i},H_{s,i},\Delta_{s,i}=0\}$. 
With equal numbers of numerator and denominator rows, Bayes' rule gives
\[
\P(\Delta_s=1\mid A_s'=a,H_s=h)
=
\frac{\omega_s(a,h)}{1+\omega_s(a,h)}.
\]
Thus, if \(\widehat e_s(a,h)\) estimates the probability on the left, then
\[
\widehat\omega_s(a,h)
=
\frac{\widehat e_s(a,h)}{1-\widehat e_s(a,h)}.
\]
Equivalently, \(\omega_s\) is characterized by the change-of-measure identity
\[
\E\{q(d_s(A_s,H_s),H_s)\}
=
\E\{\omega_s(A_s,H_s)q(A_s,H_s)\}
\]
for every square-integrable \(q\). Riesz or Bregman density-ratio learning may
therefore be used instead of classification
\citep{chernozhukov2018double,hines2025learning}.

\subsubsection{Full-history adjustment}

The full-history gradient uses the two density ratios in
\eqref{eq:fh-omega-t}. Conditional on
\(\{J_s^{(k)}=1,F_s^{(k)}=f_s\}\), the first compares the treatment law induced
by the visit-time modification \(d_s(\cdot,1,f_s)\) with the natural treatment
law, whereas the second compares the law induced by the stochastic no-visit
continuation rule with the same natural treatment law. Both ratios can be
estimated using the same classification construction as above, restricting
the training sample to \(J_s^{(k)}=1\).

For \(\omega_s^d\), the numerator treatment is
\(d_s(A_s,1,F_s^{(k)})\). For \(\omega_s^t\), the numerator treatment is $
t_s\{\alpha_s(H_s^{(k)}),\epsilon_s\}$, $\epsilon_s\sim Q_s$,
where \(Q_s\) is known. A draw from \(Q_s\) may therefore be generated for each
training observation to construct the numerator sample. Multiple randomizer
draws may also be used, with the classification odds adjusted for the resulting
class sampling proportions. The estimated ratios are then combined with the
selected bridge functions through the recursion in \eqref{eq:fh-Omega}.

\subsection{Bridge and adjoint functions}

The bridge-estimation problems in the two sections have the same local form.
For the observed-history analysis at outcome time \(t+1\), let
\[
B_t=(A_t,H_t),
\qquad
V_t=(H_t,C_{t+1}),
\qquad
M_t=R_{t+1},
\]
and use the full sample. For the full-history analysis at time \(s\), restrict
the estimation sample to \(J_s^{(k)}=1\) and let
\[
B_s=(A_s,F_s^{(k)}),
\qquad
V_s=(F_s^{(k)},C_{s+1}),
\qquad
M_s=R_{s+1}.
\]
In either case, with expectations understood under the corresponding
stratification, the bridge equation is
\[
\E\{M_s\beta_s(V_s)-1\mid B_s\}=0.
\]
The corresponding adjoint equation can be written as
\[
\E\{M_s[\lambda_s(B_s)-\phi_s(V_s)]\mid V_s\}=0.
\]
For observed-history adjustment, \((\beta_s,\lambda_s,\phi_s)\) in this generic
notation correspond to the terminal objects
\((\beta_t,\lambda_t,\phi_t)\). For full-history adjustment they correspond to
\((\beta_s,\lambda_{t+1,s},\phi_{t+1,s})\), where the target-specific loading is
defined in \eqref{eq:fh-phi-loading}. The latter can itself be estimated by a
regression of \(\Omega_s\nu_{t+1,s+1}\) on
\((F_s^{(k)},C_{s+1})\) among \(K_{s+1}^{(k)}=1\). The observed-history loading
can be estimated analogously from its conditional-expectation definition in
\Cref{ass:lambda_exist}.

This formulation places both bridge problems directly in the conditional
moment framework, without requiring preliminary estimation of a response
probability followed by inversion or an odds transformation. Small response
probabilities can still make the inverse problem weakly informative, but they
do not enter the bridge estimating equation through an explicitly estimated
inverse probability.

\subsubsection{Existing methods for conditional moment restrictions}

There is a large literature on estimating unknown functions defined by
conditional moment restrictions. Classical econometric approaches use series or
sieve approximations to convert the infinite-dimensional problem into a growing
collection of finite-dimensional moment or minimum-distance problems. In the
nonparametric instrumental-variable setting, \citet{newey2003instrumental}
propose a series-based nonparametric two-stage least-squares estimator and study
the ill-posed inverse problem that results. More generally,
\citet{ai2003efficient} develop sieve minimum-distance estimation for models
containing finite-dimensional parameters and unknown functions defined through
conditional moment restrictions. \citet{chen2012estimation} extend this line of
work to penalized sieve minimum-distance procedures for general nonparametric
conditional moment models, including settings with nonsmooth generalized
residuals. These approaches are directly applicable here after substituting the
bridge or adjoint residual for the generalized residual, although their practical
performance depends on choices of sieve dimension, basis, weighting, and
regularization.

A second family uses reproducing-kernel methods to measure violations of the
conditional moments without choosing a finite set of test functions.
\citet{muandet2020kernel} introduce the maximum moment restriction criterion,
which represents a conditional moment restriction through its largest
unconditional moment over the unit ball of an RKHS. Under suitable richness
conditions on the RKHS, the resulting population criterion characterizes the
original conditional moment restriction and has a convenient kernel
representation. This provides both a direct estimation strategy and, as used
below, a common out-of-sample criterion for comparing bridge estimators trained
by otherwise unrelated algorithms.

A third family formulates conditional-moment estimation as a minimax or
adversarial problem in which a candidate structural function competes against a
critic that searches for violated moments. \citet{lewis2018adversarial} propose
adversarial GMM for conditional moment models using flexible adversaries, and
\citet{dikkala2020minimax} develop a general minimax theory with rates for broad
hypothesis and test-function classes, including RKHSs, sparse linear classes,
random forests, and neural networks. Closely related variational formulations
use the weighting ideas of optimally weighted GMM while retaining flexible
function classes. \citet{bennett2019deepgmm} develop DeepGMM for instrumental
variable models, and \citet{bennett2023vmm} develop the variational method of
moments for general conditional moment problems, with kernel and neural-network
implementations and conditions for consistency, asymptotic normality, and
semiparametric efficiency.

Several methods have also been developed specifically for bridge-function
problems arising in proximal causal inference and related settings.
\citet{kallus2021causal} propose minimax estimators for proximal bridge
functions using general function classes such as RKHSs and neural networks and
study settings in which uniqueness of the bridge need not be imposed.
\citet{kompa2022deep} develop a deep-learning maximum-moment method for proximal
bridge estimation. Particularly close to the present setting,
\citet{bennett2023strongly} study inference on functionals of nuisance functions
defined by conditional moment restrictions when the nuisance itself can be
weakly identified or have multiple solutions; their examples explicitly include
proximal causal inference and missing-not-at-random models with shadow
variables, and they propose penalized minimax estimators for both the primary
and debiasing nuisance functions. Related work by \citet{bennett2023minimaxiv}
constructs penalized minimax estimators that converge to a selected solution of
nonparametric instrumental-variable equations even when the solution is not
unique.

These methods differ in their regularity conditions, function classes,
regularization, and computational demands, and no single approach is uniformly
preferred. We therefore treat the bridge and adjoint estimation 
problem using ensemble learners. Any candidate method that targets the corresponding conditional moment
restriction may be used, with its tuning parameters selected within the
training data. 

\subsubsection{Cross-validated convex combinations}

The linearity of both the bridge and adjoint equations makes convex stacking
particularly natural. Suppose that, for a fixed bridge problem, a library of
candidate estimators $\widehat\beta_s^{(1)},\ldots,\widehat\beta_s^{(M)}$ 
has been fit using any of the methods described above. For
\(\alpha=(\alpha_1,\ldots,\alpha_M)\) in the simplex
$\Delta_M
=
\left\{
\alpha_m\geq0:\sum_{m=1}^M\alpha_m=1
\right\}$,
define $
\widehat\beta_{s,\alpha}(v)
=
\sum_{m=1}^M\alpha_m\widehat\beta_s^{(m)}(v)$. If every candidate is an exact solution of the same population bridge equation,
then every such convex combination is also an exact solution. Thus the ensemble
construction is compatible with the identification results in both sections,
which require the target to be invariant across bridge solutions rather than
requiring the bridge itself to be unique. If the candidate solutions are
constrained to be positive, convex combination also preserves positivity.

The weights can be chosen by cross-validation using any common criterion that
measures violation of the conditional moment restriction. A convenient choice
is the RKHS maximum-moment criterion of \citet{muandet2020kernel}. Importantly,
the RKHS is used here only to score candidate estimators; the candidates
need not themselves be kernel estimators. Let \(\mathcal I_s\) denote a set of
held-out observations used to select the stacking weights, let
\(n_s=|\mathcal I_s|\), and define
\[
\widehat r_{im}
=
M_{s,i}\widehat\beta_s^{(m)}(V_{s,i})-1,
\qquad
K_{ij}^{B}=k_B(B_{s,i},B_{s,j}).
\]
When \(M_{s,i}=0\), the first term in \(\widehat r_{im}\) is zero and
\(V_{s,i}\) need not be evaluated.
For \(m,m'=1,\ldots,M\), set
\[
\widehat G_{mm'}
=
\frac{1}{n_s^2}
\sum_{i,j\in\mathcal I_s}
\widehat r_{im}\widehat r_{jm'}K_{ij}^{B}.
\]
The stacking weights are then
\begin{equation}
\widehat\alpha
\in
\arg\min_{\alpha\in\Delta_M}
\alpha^\top\widehat G\alpha.
\label{eq:bridge-stacking}
\end{equation}
Because \(k_B\) is positive semidefinite, \(\widehat G\) is positive
semidefinite, so \eqref{eq:bridge-stacking} is a small convex quadratic program.
Nonuniqueness of the optimal weights is unproblematic when several library elements converge to different exact bridge
solutions, because their convex combinations remain solutions of the same
bridge equation.

The same construction applies to the adjoint equation. Candidate adjoint
estimators should be fit using a common training-sample estimate
\(\widehat\phi_s\). To choose the stacking weights for the adjoint estimators, we use
only validation observations with $M_s=1$, since the criterion
requires $V_s$, which is unavailable when $M_s=0$. On this subset define
\[
\widehat a_{im}
=
\widehat\lambda_s^{(m)}(B_{s,i})-
\widehat\phi_s(V_{s,i}),
\]
and replace \(\widehat r_{im}\) and \(k_B\) above by
\(\widehat a_{im}\) and a kernel \(k_V\) on \(V_s\). This scores the
equivalent complete-case conditional moment restriction. The simplex constraint
is again useful: it ensures that the residual for the convexly combined adjoint
is the same convex combination of the candidate residuals. If every candidate
solves the population adjoint equation for the same \(\phi_s\), their convex
combination also solves it.

To preserve the outer cross-fitting used by the final causal estimator, neither
the library fits nor the stacking weights used on an outer validation fold may
depend on that fold. In practice, the weights can be selected by an inner
cross-validation performed entirely within each outer training sample; the
candidate learners are then refit on the full outer training sample and combined
using the selected weights before evaluation on the outer validation fold. This
is a Super-Learner-type construction \citep{vanderlaan2007super}, but here the
common validation loss is constructed from the conditional moment restriction. 

\subsection{Cross-fitting and implementation}

All nuisance functions entering a score evaluated on a validation observation
should be estimated using data that exclude that observation. For each fold,
the nuisance functions are therefore fitted on the corresponding training
sample and then evaluated on the held-out fold. The same sample split should be
propagated through the entire backward recursion so that estimated
pseudo-outcomes are constructed only from nuisance fits obtained outside their
validation fold.

Within each training fold, a practical order of operations is:
\begin{enumerate}
\item estimate the policy density ratios and the bridge solutions;
\item estimate the pooled sequential regressions, proceeding backward across
      diagonals as described in Sections~\ref{sec:estima-vd} and
      \ref{sec:estima-fh};
\item construct the target-specific loadings \(\phi_t\) or
      \(\phi_{t+1,s}\) using the fitted density ratios, bridge functions, and
      sequential regressions;
\item estimate the corresponding adjoint functions \(\lambda_t\) or
      \(\lambda_{t+1,s}\);
\item evaluate the fold-specific sequential and bridge gradient components and
      construct the estimators in \eqref{eq:vd-one-step} and
      \eqref{eq:fh-one-step}.
\end{enumerate}
The tuning parameters used to select bridge and adjoint representatives must
also be chosen within the training folds. Under the product-rate conditions in
Propositions~\ref{prop:vd-second-order} and \ref{prop:fh-second-order}, inference
may be based on the empirical covariance of the cross-fitted estimated
gradients, with simultaneous inference across outcome times implemented as
described in \S\ref{sec:estima-vd}.

%% file: conclusion.tex

\section{Conclusion}

We have developed identification and estimation methods for longitudinal causal effects under informative missingness governed by a self-censoring model. The motivating application is the analysis of electronic health records, in which treatments can be modified only during healthcare encounters, while the occurrence of these encounters depends on covariates and outcomes that may otherwise remain unobserved. Our framework targets interventions on treatment while allowing the visit process to evolve naturally under the intervention, i.e., it avoids counterfactual variables that inquire what would have happened under an intervention on the visit process. We consider both settings in which the observed history suffices for treatment confounding control and settings in which adjustment requires the latent covariate history.

The identification results combine longitudinal modified treatment policies with shadow-variable conditional moment restrictions. A central feature is that identification of the causal parameter does not require uniqueness of the bridge functions. Instead, target-specific adjoint range conditions ensure that different bridge solutions yield the same counterfactual mean. These conditions also play a central role in the estimation theory: adjoint solutions provide the corrections needed to account for estimation of the bridge functions. The resulting gradients and second-order expansions support estimators that accommodate machine learning and permit root-n inference under suitable nuisance-estimation conditions. Estimation based on these gradients does not necessarily attain the semiparametric efficiency bound, which requires projection onto the tangent space of the observed-data model. Our results provide a foundation for constructing semiparametrically efficient estimators, a task we leave for future work.


Future work includes additional more extensive evaluations of the finite-sample performance of the proposed estimators, including the choice of bridge and adjoint learners, regularization, and pooling across outcome times. Sensitivity analyses for violations of the shadow-variable restrictions are also important directions. The results presented here provide a foundation for these developments and for causal analyses of longitudinal data in which measurement encounters serve simultaneously as opportunities for treatment and informative measurement events, such as healthcare encounters.

%% file: observed_proofs.tex
\section{Proofs of main results for visit-dependent exposures}\label{app:main-proofs}
\subsection{Auxiliary results}\label{app:auxiliary-results}

\begin{lemma}[Inverse-probability representation for visit-dependent exposures]\label{lem:nor_identity}
Under Assumption \ref{ass:positive}, for any square-integrable function \(q(H_t,C_{t+1})\),
\[
\E\{q(H_t,C_{t+1})\}
=
\E\!\left[
R_{t+1}\beta_t^0(H_t,C_{t+1})q(H_t,C_{t+1})
\right].
\]
\end{lemma}
\begin{proof}
By iterated expectation and the definition of the true inverse-probability bridge,
\begin{align*}
\E\!\left[
R_{t+1}\beta_t^0(H_t,C_{t+1})q(H_t,C_{t+1})
\right]
&=
\E\!\left[
\P(R_{t+1}=1\mid H_t,C_{t+1})
\beta_t^0(H_t,C_{t+1})q(H_t,C_{t+1})
\right]\\
&=
\E\{q(H_t,C_{t+1})\}.
\end{align*}
\end{proof}
\begin{lemma}[Conditional change of measure]
\label{lem:conditional-change-of-measure}
Let $(X,W)$ be random variables with conditional density
$g(x\mid w)$ for $X$ given $W=w$, with respect to a dominating measure
$\nu$. Let $g^\star(x\mid w)$ be another conditional density such that
$g^\star(\cdot\mid w)$ is absolutely continuous with respect to
$g(\cdot\mid w)$ almost surely, and define
\[
\kappa(x,w)
:=
\frac{g^\star(x\mid w)}{g(x\mid w)}.
\]
Let $X^\star$ denote a random variable whose conditional density given
$W=w$ is $g^\star(\cdot\mid w)$. Then, for every measurable function
$h$ for which the expectations below are finite,
\[
\E\{h(X^\star,W)\}
=
\E\{\kappa(X,W)h(X,W)\}.
\]
More generally, if $L$ is an integrable random variable measurable with
respect to a sigma-field containing $W$ and satisfying
\[
\E(L\mid X,W)
=
\E(L\mid W)
\]
almost surely, then
\[
\E\{Lh(X^\star,W)\}
=
\E\{L\kappa(X,W)h(X,W)\}.
\]
\end{lemma}

\begin{proof}
By iterated expectation,
\begin{align*}
\E\{h(X^\star,W)\}
&=
\E\!\left[
\E\{h(X^\star,W)\mid W\}
\right]\\
&=
\E\!\left[
\int h(x,W)g^\star(x\mid W)\,d\nu(x)
\right].
\end{align*}
By absolute continuity,
\[
g^\star(x\mid W)
=
\kappa(x,W)g(x\mid W)
\]
for $g(\cdot\mid W)$-almost every $x$. Therefore,
\begin{align*}
\E\{h(X^\star,W)\}
&=
\E\!\left[
\int
h(x,W)\kappa(x,W)g(x\mid W)\,d\nu(x)
\right]\\
&=
\E\!\left[
\E\{\kappa(X,W)h(X,W)\mid W\}
\right]\\
&=
\E\{\kappa(X,W)h(X,W)\}.
\end{align*}

For the weighted identity, first use iterated expectation and the assumption
$\E(L\mid X,W)=\E(L\mid W)$ to obtain
\begin{align*}
\E\{L\kappa(X,W)h(X,W)\}
&=
\E\!\left[
\E\{L\kappa(X,W)h(X,W)\mid X,W\}
\right]\\
&=
\E\!\left[
\E(L\mid W)\kappa(X,W)h(X,W)
\right]\\
&=
\E\!\left[
\E(L\mid W)
\E\{\kappa(X,W)h(X,W)\mid W\}
\right].
\end{align*}
By the first part of the lemma,
\[
\E\{\kappa(X,W)h(X,W)\mid W\}
=
\E\{h(X^\star,W)\mid W\}.
\]
Hence,
\begin{align*}
\E\{L\kappa(X,W)h(X,W)\}
&=
\E\!\left[
\E(L\mid W)
\E\{h(X^\star,W)\mid W\}
\right]\\
&=
\E\{Lh(X^\star,W)\},
\end{align*}
where the final equality follows from the conditional independence encoded
by $\E(L\mid X,W)=\E(L\mid W)$.
\end{proof}

\subsection{Range conditions in finite discrete models}\label{app:range-discrete}

This section gives a finite-dimensional interpretation of the bridge and adjoint range conditions in Assumption~\ref{ass:lambda_exist}. We clarify that the adjoint condition is weaker than uniqueness of the bridge, and show examples where the adjoint range condition fails or is highly restrictive. We focus on the visit-dependent case as the examples translate with few modifications to the general case.

Suppose the shadow variable $A_t$ takes values $a_1,\ldots,a_r$ and the latent variable $C_{t+1}$ takes values $c_1,\ldots,c_m$. For fixed $h_t$, let $T_{t,h_t}$ denote the $r\times m$ matrix with entries
\[
(T_{t,h_t})_{ij}
=
\P(R_{t+1}=1,C_{t+1}=c_j\mid A_t=a_i,H_t=h_t).
\]
For a vector $b_{t,h_t}=(b_{t,h_t}(c_1),\ldots,b_{t,h_t}(c_m))^\top$, the bridge equation is the linear system
\[
T_{t,h_t}b_{t,h_t}=\boldsymbol 1_r,
\]
where $\boldsymbol 1_r$ is the $r$-vector of ones. Thus existence of a bridge at $h_t$ is the ordinary linear-algebra condition $\boldsymbol 1_r\in\mathrm{Range}(T_{t,h_t})$.

The adjoint in Assumption~\ref{ass:lambda_exist} is the transpose of this matrix after accounting for the inner-product weights. Let
\[
W_{A,t,h_t}=\operatorname{diag}\{\P(A_t=a_i\mid H_t=h_t):i=1,\ldots,r\},
\]
and
\[
W_{C,t,h_t}=\operatorname{diag}\{\P(R_{t+1}=1,C_{t+1}=c_j\mid H_t=h_t):j=1,\ldots,m\}.
\]
With the weighted $L^2$ inner products used in the main text, the adjoint matrix is equal to
\[
T_{t,h_t}^\star=W_{C,t,h_t}^{-1}T_{t,h_t}^\top W_{A,t,h_t}.
\]
Equivalently, up to nonsingular weighting, $T_{t,h_t}^\star$ is the transpose of $T_{t,h_t}$. Therefore
$
\mathrm{Range}(T_{t,h_t}^\star)=\mathrm{Null}(T_{t,h_t})^\perp$,
where orthogonality is with respect to the measure $\P(R_{t+1}=1,C_{t+1}\in\cdot\mid H_t=h_t)$. The target-specific range condition $\phi_{t,h_t}\in\mathrm{Range}(T_{t,h_t}^\star)$, where
\[
\phi_{t,h_t}=(\phi_t(h_t,c_1),\ldots,\phi_t(h_t,c_m))^\top,
\]
is therefore equivalent to
\[
\phi_{t,h_t}^\top W_{C,t,h_t}f_{t,h_t}=0
\quad\text{for every }f_{t,h_t}\text{ such that }T_{t,h_t}f_{t,h_t}=0.
\]
This expression makes the role of the adjoint assumption transparent. If $b_{t,h_t}$ and $b_{t,h_t}+f_{t,h_t}$ are two bridge solutions, then they give the same target contribution precisely when $\phi_{t,h_t}$ is orthogonal to every null direction $f_{t,h_t}$ of the bridge operator. Thus the bridge may be nonunique, while the scalar target remains unique. In particular, the pointwise bridge-based target contribution is
\[
\theta_{t,h_t}(b)
=
\phi_{t,h_t}^{\top}W_{C,t,h_t}b_{t,h_t}.
\]
If $\widetilde b_{t,h_t}=b_{t,h_t}+f_{t,h_t}$ is another solution, their difference is
\[
\theta_{t,h_t}(\widetilde b)-\theta_{t,h_t}(b)
=
\phi_{t,h_t}^{\top}W_{C,t,h_t}f_{t,h_t}.
\]
Because both bridges solve the same equation,
$T_{t,h_t}\widetilde b_{t,h_t}=T_{t,h_t}b_{t,h_t}=\boldsymbol 1_r$, so $T_{t,h_t}f_{t,h_t}=0$.

\paragraph{Example 1: full-rank bridge, range condition holds automatically.}
Suppose $A_t$ and $C_{t+1}$ are both binary and, for a fixed $h_t$,
\[
T_{t,h_t}=
\begin{pmatrix}
0.4 & 0.1\\
0.15 & 0.35
\end{pmatrix}.
\]
This matrix has rank two. Hence $T_{t,h_t}$ is invertible, the bridge equation $T_{t,h_t}b_{t,h_t}=\boldsymbol 1_2$ has a unique solution, and $T_{t,h_t}^\star$ also has full range. Consequently every target-relevant function $\phi_{t,h_t}$ belongs to $\mathrm{Range}(T_{t,h_t}^\star)$. In this case the shadow variable is sufficiently informative to identify the bridge itself.

\paragraph{Example 2: nonunique bridge, but the target is still identified.}
Suppose $A_t$ is binary and $C_{t+1}$ has three levels, and
\[
T_{t,h_t}=
\begin{pmatrix}
1/4 & 1/4 & 0\\
0 & 1/4 & 1/4
\end{pmatrix}.
\]
The matrix has rank two, so the bridge equation has solutions, but the solution is not unique because
\[
h_{t,h_t}=(1,-1,1)^\top
\quad\text{satisfies}\quad
T_{t,h_t}h_{t,h_t}=0.
\]
If the target-relevant function satisfies, in the corresponding bridge-side inner product,
\[
\phi_{t,h_t}^\top W_{C,t,h_t}h_{t,h_t}=0,
\]
then $\phi_{t,h_t}\in\mathrm{Range}(T_{t,h_t}^\star)$ and the target contribution is unique even though the bridge is not. For example, if
$\P(A_t=a_1\mid H_t=h_t)=\P(A_t=a_2\mid H_t=h_t)=1/2$, then
\[
W_{C,t,h_t}=\operatorname{diag}(1/8,1/4,1/8),
\]
and the vector
\[
\phi_{t,h_t}=(1,2,3)^\top
\]
satisfies $\phi_{t,h_t}^\top W_{C,t,h_t}h_{t,h_t}=0$. Thus replacing a bridge solution $b_{t,h_t}$ by $b_{t,h_t}+c h_{t,h_t}$ changes the bridge but leaves the target contribution unchanged.

\paragraph{Example 3: bridge exists, but the target-specific range condition fails.}
Use the same matrix and treatment probabilities as in Example 2, so the bridge is nonunique with null direction $h_{t,h_t}=(1,-1,1)^\top$ and
$W_{C,t,h_t}=\operatorname{diag}(1/8,1/4,1/8)$. Now take
\[
\phi_{t,h_t}=(1,0,0)^\top.
\]
Then $\phi_{t,h_t}^\top W_{C,t,h_t}h_{t,h_t}=1/8\neq0$, so $\phi_{t,h_t}\notin\mathrm{Range}(T_{t,h_t}^\star)$. If $b_{t,h_t}$ solves $T_{t,h_t}b_{t,h_t}=\boldsymbol 1_2$, then so does $b_{t,h_t}+c h_{t,h_t}$ for every scalar $c$, but the resulting target contribution changes with $c$. The observed data therefore do not determine the target contribution.

\paragraph{Example 4: failure of the bridge range condition.}
Suppose $A_t$ and $C_{t+1}$ are binary and the distribution of $C_{t+1}$ among complete cases is the same across levels of the shadow variable, but the visit probabilities differ:
\[
T_{t,h_t}=
\begin{pmatrix}
1/4 & 1/4\\
1/8 & 1/8
\end{pmatrix}.
\]
Then $\mathrm{Range}(T_{t,h_t})$ consists only of vectors proportional to $(1,1/2)^\top$, so $\boldsymbol 1_2\notin\mathrm{Range}(T_{t,h_t})$ and there is no bridge solution. Intuitively, if the complete-case distribution of the latent variable does not vary with the shadow variable, no function of that latent variable can offset different visit probabilities across levels of the shadow variable.

\paragraph{Example 5: binary shadow variable and arbitrary latent variable.}
The previous examples are finite dimensional, but the same geometry explains why the adjoint range condition can be restrictive when the shadow variable is low-dimensional and the latent variable has larger dimension. Suppose, for a fixed history value $h_t$, that $A_t\in\{0,1\}$. Then the bridge operator maps a function $b_{t,h_t}(c)$ to the two conditional means
\[
(T_{t,h_t}b_{t,h_t})(a)=
\E\{R_{t+1}b_{t,h_t}(C_{t+1})\mid A_t=a,H_t=h_t\},
\qquad a\in\{0,1\}.
\]
Thus the bridge equation imposes only two moment restrictions on the function $b_{t,h_t}$, and its null space is infinite dimensional. The adjoint maps a two-vector $\lambda_{t,h_t}=(\lambda_t(0,h_t),\lambda_t(1,h_t))^\top$ to
\[
(T_{t,h_t}^\star\lambda_{t,h_t})(c)
=
\E\{\lambda_t(A_t,h_t)\mid R_{t+1}=1,C_{t+1}=c,H_t=h_t\}.
\]
Equivalently, letting
\[
e_{t,h_t}(c):=\P(A_t=1\mid R_{t+1}=1,C_{t+1}=c,H_t=h_t),
\]
the range of the adjoint is
\[
\mathrm{Range}(T_{t,h_t}^\star)
=
\left\{
\lambda_t(0,h_t)\{1-e_{t,h_t}(c)\}+\lambda_t(1,h_t)e_{t,h_t}(c):
\lambda_t(0,h_t),\lambda_t(1,h_t)\in\mathbb R
\right\}.
\]
This is the space spanned by the functions $f(c)=1$ and $e_{t,h_t}(c)$. Denote the target-relevant function $\phi_t(h_t,c)$ as$ \phi_{t,h_t}(c)$. Consider as an example the intervention $d_t(a_t,h_t)=1$. The corresponding density ratio is
\[
\omega_t(A_t,h_t)
=
\frac{I(A_t=1)}{g_t(1\mid h_t)}.
\]
Since $Y_{t+1}=h(C_{t+1})$, the definition of $\phi_t$ gives
\begin{align*}
\phi_{t, h_t}(c)
&=
\Omega_{t-1}\,h(c)\,
\E\!\left\{
\frac{I(A_t=1)}{g_t(1\mid h_t)}
\mathrel{\Big|}
R_{t+1}=1,H_t=h_t,C_{t+1}=c
\right\}\\
&=
\frac{\Omega_{t-1}}{g_t(1\mid h_t)}
\,h(c)e_{t,h_t}(c),
\end{align*}
where $\Omega_{t-1}$ is evaluated at the fixed history $h_t$. Therefore, for this static intervention, the adjoint range condition is equivalent to requiring
\[
h(c)e_{t,h_t}(c)
\in
\operatorname{span}\{1,e_{t,h_t}(c)\}.
\]
This implies that the treatment probabilities must satisfy
\[e_{t, h_t}(c) = \frac{\alpha_{t,h_t}}{h(c) - \delta_{t,h_t}} \]
for constants $\alpha_{t,h_t}$ and $\delta_{t,h_t}$. 
 
\subsection{Proof of Theorem~\ref{thm:latent_bridge}}
\begin{proof}
By the shadow-variable assumption,
\[
\P(R_{t+1}=1\mid A_t,H_t,C_{t+1})
=
\P(R_{t+1}=1\mid H_t,C_{t+1}).
\]
Therefore, for every $(a_t,h_t)$ in the relevant support,
\begin{align*}
(T_t\beta_t^0)(a_t,h_t)
&=
\E\!\left[
R_{t+1}\beta_t^0(H_t,C_{t+1})
\mid A_t=a_t,H_t=h_t
\right]\\
&=
\E\!\left[
\P(R_{t+1}=1\mid A_t,H_t,C_{t+1})
\beta_t^0(H_t,C_{t+1})
\mid A_t=a_t,H_t=h_t
\right]\\
&=1.
\end{align*}
The same calculation with $Y_{t+1}$ included gives
\begin{align*}
\E\!\left[
R_{t+1}Y_{t+1}\beta_t^0(H_t,C_{t+1})
\mid A_t=a_t,H_t=h_t
\right]
&=
\E\!\left[Y_{t+1}\mid A_t=a_t,H_t=h_t\right]\\
&=m_{t+1,t}(a_t,h_t),
\end{align*}
which proves \eqref{eq:m_bridge_repr} and \eqref{eq:tildeY_def}.
\end{proof}

\subsection{Proof of Theorem~\ref{thm:theta_identified}}
\begin{proof}
By the no unmeasured confounding and support assumptions, the usual MTP identification argument \citep{diaz2023nonparametric} gives
\[
\theta_{t+1}(\bar d_t)=\E\{\Omega_tY_{t+1}\}.
\]
Because $Y_{t+1}=h(C_{t+1})$, iterated expectation and the shadow-variable assumption give
\begin{align*}
&\E\{R_{t+1}\Omega_tY_{t+1}\beta_t^0(H_t,C_{t+1})\}\\
&\quad=
\E\!\left[
\Omega_tY_{t+1}\beta_t^0(H_t,C_{t+1})
\P(R_{t+1}=1\mid A_t,H_t,C_{t+1})
\right]\\
&\quad=\E\{\Omega_tY_{t+1}\},
\end{align*}
so, for the true bridge,
\[
\theta_{t+1}(\bar d_t)
=
\E\{R_{t+1}\Omega_tY_{t+1}\beta_t^0(H_t,C_{t+1})\}.
\]
Now let $\beta_t$ be any square-integrable solution of \eqref{eq:fredholm_wor} and write $\Delta_t=\beta_t-\beta_t^0$. Since both solve the same bridge equation,
\[
T_t\Delta_t=0.
\]
By iterated expectation and the definition of $\phi_t$ and \Cref{ass:lambda_exist},
\begin{align*}
&\E\{R_{t+1}\Omega_tY_{t+1}\Delta_t(H_t,C_{t+1})\}\\
&\qquad=
\E\{R_{t+1}\phi_t(H_t,C_{t+1})\Delta_t(H_t,C_{t+1})\}\\
&\qquad=
\langle \Delta_t,T_t^\star\lambda_t\rangle_{C,t}
=
\langle T_t\Delta_t,\lambda_t\rangle_{A,t}
=0.
\end{align*}
Hence
\[
\theta_{t+1}(\bar d_t)
=
\E\{R_{t+1}\Omega_tY_{t+1}\beta_t(H_t,C_{t+1})\}.
\]
Finally,
\begin{align*}
\E\{R_{t+1}\Omega_tY_{t+1}\beta_t(H_t,C_{t+1})\}
&=
\langle \beta_t,\phi_t\rangle_{C,t}\\
&=
\langle T_t\beta_t,\lambda_t\rangle_{A,t}\\
&=
\E\{\lambda_t(A_t,H_t)\},
\end{align*}
because $T_t\beta_t=1$. If $\lambda_t$ and $\widetilde\lambda_t$ are two solutions of \eqref{eq:rep-iden}, then with $\Delta_\lambda=\lambda_t-\widetilde\lambda_t$,
\[
\E(\Delta_\lambda)
=
\langle T_t\beta_t,\Delta_\lambda\rangle_{A,t}
=
\langle \beta_t,T_t^\star\Delta_\lambda\rangle_{C,t}
=0,
\]
so the identified value does not depend on the chosen representer.
\end{proof}

\subsection{Proof of Corollary~\ref{cor:vd-sequential-regression}}
\begin{proof}
Fix any square-integrable solution $\beta_t$ of \eqref{eq:fredholm_wor} and let $\eta_t$ be the corresponding pseudo-outcome. By the usual sequential-regression representation for longitudinal modified treatment policies \citep{diaz2023nonparametric}, applied with $\eta_t$ as the terminal outcome held fixed,
\[
\E\!\left[
m_{t+1,1}\{d_1(A_1,H_1),H_1\}
\right]
=
\E\{\Omega_t\eta_t(O)\}.
\]
By the definition of $\eta_t$ and Theorem~\ref{thm:theta_identified},
\[
\E\{\Omega_t\eta_t(O)\}
=
\E\{R_{t+1}\Omega_tY_{t+1}\beta_t(H_t,C_{t+1})\}
=
\theta_{t+1}(\bar d_t).
\]
The invariance to the particular bridge solution follows from the same null-space/adjoint calculation in the proof of Theorem~\ref{thm:theta_identified}.
\end{proof}

\subsection{Proof of Theorem~\ref{thm:vd-if}}\label{proof:thm:vd-if}
\begin{assumption}[Regularity of the selected visit-dependent bridge path]
\label{ass:vd-beta-path-regularity}
Along every regular one-dimensional submodel $\{P_\varepsilon\}$ through $P$,
select the minimum-norm square-integrable solution $\beta_{t,\varepsilon}$ of
$T_{t,\varepsilon}\beta_{t,\varepsilon}=1$. Assume there is
$\dot\beta_t\in L^2(P)$ such that
\[
\|\beta_{t,\varepsilon}-\beta_{t,0}-\varepsilon\dot\beta_t\|_{L^2(P)}
=o(|\varepsilon|).
\]
\end{assumption}
The assumption selects a canonical bridge along each regular parametric
submodel and requires that this selected solution be differentiable along the
submodel.

\begin{proof}
\medskip\noindent\textbf{Part I: Gradient calculation.}\par
Fix a regular parametric submodel $\{P_\varepsilon:\varepsilon\in(-\delta,\delta)\}$
through $P$, and let
\[
S(O)=\left.\frac{d}{d\varepsilon}\log p_\varepsilon(O)\right|_{\varepsilon=0}
\]
denote its score. A dot denotes differentiation at $\varepsilon=0$.

First hold the terminal pseudo-outcome
$\eta_{t+1,t+1}=R_{t+1}Y_{t+1}\beta_t$ fixed. Differentiating the ordinary
longitudinal modified-policy recursion then gives
\[
\left.\frac{d}{d\varepsilon}\theta_{t+1}(\bar d_t;P_\varepsilon)
\right|_{\varepsilon=0,\,\beta_t\text{ fixed}}
=
\E\{D_t^{\mathrm{seq}}(O)S(O)\}.
\]
This is the standard longitudinal modified-policy calculation: at each time
$s$, the derivative of the conditional regression contributes the residual
$\eta_{t+1,s+1}-m_{t+1,s}$ multiplied by the likelihood ratio $\Omega_s$, and
the derivative of the baseline covariate distribution contributes
$q_{t+1,1}-\theta_{t+1}$. See \cite{diaz2023nonparametric}.

It remains to differentiate the dependence of the terminal pseudo-outcome on
the bridge. Its bridge-dependent derivative is
\[
\dot\eta_{t+1,t+1}(O)
=
R_{t+1}Y_{t+1}\dot\beta_t(H_t,C_{t+1}).
\]
Using the re-weighted formulation of the identification result gives the contribution
\begin{align*}
\E\{\Omega_tR_{t+1}Y_{t+1}\dot\beta_t(H_t,C_{t+1})\}.
\end{align*}
By the definition of $\phi_t$ and iterated expectation conditional on
$(R_{t+1}=1,H_t,C_{t+1})$,
\begin{align*}
&\E\{\Omega_tR_{t+1}Y_{t+1}h(H_t,C_{t+1})\}\\
&\quad=
\E\!\left[
R_{t+1}h(H_t,C_{t+1})
\E\{\Omega_tY_{t+1}\mid R_{t+1}=1,H_t,C_{t+1}\}
\right]\\
&\quad=
\E\{R_{t+1}\phi_t(H_t,C_{t+1})h(H_t,C_{t+1})\}
\end{align*}
for every square-integrable $h$. Since
$T_t^\star\lambda_t=\phi_t$, the adjoint identity gives
\begin{align*}
&\E\{R_{t+1}\phi_t(H_t,C_{t+1})h(H_t,C_{t+1})\}\\
&\quad=
\E\{\lambda_t(A_t,H_t)(T_th)(A_t,H_t)\}\\
&\quad=
\E\{R_{t+1}\lambda_t(A_t,H_t)h(H_t,C_{t+1})\}.
\end{align*}
Taking $h=\dot\beta_t$ therefore converts the bridge derivative into
\[
\E\{R_{t+1}\lambda_t(A_t,H_t)\dot\beta_t(H_t,C_{t+1})\}.
\]
The bridge equation is equivalent to the unconditional moment restriction
\[
\E\!\left[
q(A_t,H_t)\{R_{t+1}\beta_t(H_t,C_{t+1})-1\}
\right]=0
\]
for every square-integrable $q(A_t,H_t)$. Differentiating this restriction along the submodel and taking $q=\lambda_t$ gives
\begin{align*}
0
&=
\E\!\left[
\lambda_t(A_t,H_t)R_{t+1}\dot\beta_t(H_t,C_{t+1})
\right]\\
&\quad+
\E\!\left[
\lambda_t(A_t,H_t)\{R_{t+1}\beta_t(H_t,C_{t+1})-1\}S(O)
\right].
\end{align*}
Consequently,
\[
\E\{R_{t+1}\lambda_t\dot\beta_t\}
=
-\E\!\left[
\lambda_t\{R_{t+1}\beta_t-1\}S(O)
\right],
\]
which is exactly the bridge contribution $\E\{D_t^{\mathrm{br}}(O)S(O)\}$.
Adding the sequential-regression part
proves
\[
\left.\frac{d}{d\varepsilon}\theta_{t+1}(\bar d_t;P_\varepsilon)
\right|_{\varepsilon=0}
=
\E\{D_t(O)S(O)\}.
\]
Thus $D_t$ is a gradient.

\end{proof}

\subsection{Proof of Proposition~\ref{prop:vd-second-order}}

\begin{proof}
For a distribution $X$ and a square-integrable function $\beta$, let
\[
\theta_{t+1}(X,\beta)
\]
denote the longitudinal sequential-regression functional under $X$ with
terminal pseudo-outcome
\[
\eta_{t+1,t+1}^{\beta}
=
R_{t+1}Y_{t+1}\beta(H_t,C_{t+1}).
\]
Let $m_{t+1,s,X}^{\beta}$ denote the corresponding sequential regressions.
In particular,
\[
\theta_{t+1}(\bar d_t;X)
=
\theta_{t+1}(X,\beta_{t,X}),
\qquad
m_{t+1,s,X}
=
m_{t+1,s,X}^{\beta_{t,X}}.
\]
For notational simplicity, write
\[
\beta_P:=\beta_{t,P},
\qquad
\beta_F:=\beta_{t,F},
\qquad
\lambda_P:=\lambda_{t,P},
\qquad
\lambda_F:=\lambda_{t,F}.
\]

We decompose
\begin{align}
&\theta_{t+1}(P,\beta_P)
-
\theta_{t+1}(F,\beta_F)
\nonumber\\
&\qquad=
\left\{
\theta_{t+1}(P,\beta_P)
-
\theta_{t+1}(P,\beta_F)
\right\}
+
\left\{
\theta_{t+1}(P,\beta_F)
-
\theta_{t+1}(F,\beta_F)
\right\}.
\label{eq:vd-rem-decomposition1}
\end{align}
The second term in \eqref{eq:vd-rem-decomposition1} is an ordinary
longitudinal modified-treatment-policy comparison in which the terminal
pseudo-outcome
\[
R_{t+1}Y_{t+1}\beta_F(H_t,C_{t+1})
\]
is held fixed. Its von Mises decomposition is therefore the standard
sequential-regression calculation of \citet{diaz2023nonparametric}. In
particular,
\begin{align}
&\theta_{t+1}(P,\beta_F)
-
\theta_{t+1}(F,\beta_F)
\nonumber\\
&\qquad=
\E_P\{D_t^{\mathrm{seq}}(O;F)\}
-
R_{t,2}^{\mathrm{seq}}(P,F;\beta_F),
\label{eq:vd-rem-standard-seq}
\end{align}
where
\begin{align}
R_{t,2}^{\mathrm{seq}}(P,F;\beta_F)
&=
\sum_{s=1}^t
\E_P\!\left[
\Omega_{s-1,F}
\{\omega_{s,P}-\omega_{s,F}\}
\{m_{t+1,s,F}^{\beta_F}
-
m_{t+1,s,P}^{\beta_F}\}
\right].
\label{eq:vd-rem-standard-seq-remainder}
\end{align}
Thus, the only calculation that differs from the standard longitudinal
argument concerns the first term in \eqref{eq:vd-rem-decomposition1}, which
captures the effect of changing the bridge function while holding the law
fixed at $P$.

By linearity of the longitudinal g-computation functional in its terminal
pseudo-outcome,
\begin{align}
&\theta_{t+1}(P,\beta_P)
-
\theta_{t+1}(P,\beta_F)
\nonumber\\
&\qquad=
\E_P\!\left[
R_{t+1}\Omega_{t,P}Y_{t+1}
\{\beta_P(H_t,C_{t+1})-\beta_F(H_t,C_{t+1})\}
\right].
\notag
\end{align}
By the definition of $\phi_{t,P}$,
\begin{align}
&\E_P\!\left[
R_{t+1}\Omega_{t,P}Y_{t+1}
(\beta_P-\beta_F)
\right]
\nonumber\\
&\qquad=
\E_P\!\left[
R_{t+1}\phi_{t,P}(H_t,C_{t+1})
(\beta_P-\beta_F)
\right].
\notag
\end{align}
Since
\[
T_{t,P}^{\star}\lambda_P=\phi_{t,P},
\]
the adjoint identity gives
\begin{align}
&\theta_{t+1}(P,\beta_P)
-
\theta_{t+1}(P,\beta_F)
\nonumber\\
&\qquad=
\left\langle
\beta_P-\beta_F,
T_{t,P}^{\star}\lambda_P
\right\rangle_{C,t,P}
\nonumber\\
&\qquad=
\left\langle
T_{t,P}(\beta_P-\beta_F),
\lambda_P
\right\rangle_{A,t,P}
\nonumber\\
&\qquad=
\E_P\!\left[
R_{t+1}\lambda_P(A_t,H_t)
\{\beta_P(H_t,C_{t+1})-\beta_F(H_t,C_{t+1})\}
\right].
\label{eq:vd-rem-beta-difference}
\end{align}
We next evaluate the bridge component of the gradient at $F$:
\[
D_t^{\mathrm{br}}(O;F)
=
-\lambda_F\{R_{t+1}\beta_F-1\}.
\]
Under $P$, the true bridge satisfies $T_{t,P}\beta_P=1$. Therefore,
\begin{align}
\E_P\{D_t^{\mathrm{br}}(O;F)\}
&=
\E_P\!\left[
\lambda_F\{1-T_{t,P}\beta_F\}
\right]\nonumber\\
&=
\E_P\!\left[
\lambda_F T_{t,P}(\beta_P-\beta_F)
\right]\nonumber\\
&=
\E_P\!\left[
R_{t+1}\lambda_F(\beta_P-\beta_F)
\right].
\label{eq:vd-rem-bridge-expectation}
\end{align}

Combining
\eqref{eq:vd-rem-decomposition1},
\eqref{eq:vd-rem-standard-seq},
\eqref{eq:vd-rem-beta-difference}, and
\eqref{eq:vd-rem-bridge-expectation} yields
\begin{align}
&\theta_{t+1}(P,\beta_P)
-
\theta_{t+1}(F,\beta_F)
\nonumber\\
&\qquad=
\E_P\{D_t^{\mathrm{seq}}(O;F)\}
+
\E_P\{D_t^{\mathrm{br}}(O;F)\}
-
R_{t,2}^{\mathrm{seq}}(P,F;\beta_F)
\nonumber\\
&\qquad\quad+
\E_P\!\left[
R_{t+1}
(\lambda_P-\lambda_F)
(\beta_P-\beta_F)
\right]
\nonumber\\
&\qquad=
\E_P\{D_t(O;F)\}
-
R_{t,2}^{\mathrm{seq}}(P,F;\beta_F)
\nonumber\\
&\qquad\quad+
\E_P\!\left[
R_{t+1}
(\lambda_P-\lambda_F)
(\beta_P-\beta_F)
\right].
\notag
\end{align}
Equivalently,
\begin{align}
&\theta_{t+1}(\bar d_t;F)
-
\theta_{t+1}(\bar d_t;P)
\nonumber\\
&\qquad=
-\E_P\{D_t(O;F)\}
+
R_{t,2}^{\mathrm{seq}}(P,F;\beta_F)
\nonumber\\
&\qquad\quad-
\E_P\!\left[
R_{t+1}
(\lambda_P-\lambda_F)
(\beta_P-\beta_F)
\right].
\notag
\end{align}
Therefore the bridge contribution to the second-order remainder satisfies,
by Cauchy--Schwarz,
\begin{align}
&\left|
\E_P\!\left[
R_{t+1}
(\lambda_P-\lambda_F)
(\beta_P-\beta_F)
\right]
\right|
\nonumber\\
&\qquad\le
\|\lambda_{t,P}-\lambda_{t,F}\|_{2,P}
\|\beta_{t,P}-\beta_{t,F}\|_{2,P}.
\notag
\end{align}

The remaining term is the usual longitudinal sequential-regression
remainder. By \eqref{eq:vd-rem-standard-seq-remainder} and the boundedness of
the cumulative policy weights,
\begin{align}
\left|
R_{t,2}^{\mathrm{seq}}(P,F;\beta_F)
\right|
&\lesssim
\sum_{s=1}^t
\|\omega_{s,P}-\omega_{s,F}\|_{2,P}
\|
m_{t+1,s,P}^{\beta_F}
-
m_{t+1,s,F}^{\beta_F}
\|_{2,P},
\notag
\end{align}
which proves the result.
\end{proof}

%% file: full_proofs.tex
\section{Proofs for full-history adjustment}\label{app:fh-proofs}

Throughout this appendix, equalities are understood on the relevant support sets. We first give the argument with the complete history and then show how \Cref{ass:fh-markov} yields the finite-\(k\) representation used in the main text.

\subsection{Proof of Lemma~\ref{lem:fh-bridge-equation}}

\begin{proof}
Fix \(s\leq t\). Since conditioning on \(J_s=1,F_s\) fixes the remote history \(\mathcal R_s^{(k)}\) and implies \(J_s^{(k)}=1\), M3 gives, on the relevant support,
\begin{align*}
&\P(R_{s+1}=1\mid J_s=1,A_s,F_s,C_{s+1})\\
&\qquad=
\P(R_{s+1}=1\mid J_s^{(k)}=1,A_s,F_s^{(k)},C_{s+1}).
\end{align*}
The left-hand side does not depend on \(A_s\) by \Cref{ass:fh-shadow-markov}. Hence the reduced conditional probability on the right does not depend on \(A_s\):\begin{equation}
\begin{split}
&\P(R_{s+1}=1\mid J_s^{(k)}=1,A_s,F_s^{(k)},C_{s+1})\\
&\qquad=
\P(R_{s+1}=1\mid J_s^{(k)}=1,F_s^{(k)},C_{s+1}).
\end{split}
\label{eq:fh-reduced-shadow}
\end{equation}
Finally,
\begin{align*}
(T_s\beta_s^0)(A_s,F_s^{(k)})
&=
\E\{R_{s+1}\beta_s^0(F_s^{(k)},C_{s+1})\mid J_s^{(k)}=1,A_s,F_s^{(k)}\}\\
&=
\E\!\left[
\beta_s^0(F_s^{(k)},C_{s+1})
\P(R_{s+1}=1\mid J_s^{(k)}=1,A_s,F_s^{(k)},C_{s+1})
\mid J_s^{(k)}=1,A_s,F_s^{(k)}
\right]\\
&=1
\end{align*}
by \eqref{eq:fh-reduced-shadow} and the definition of \(\beta_s^0\).
\end{proof}

\subsection{Auxiliary variables}
We will now define some auxiliary variables that will be used in the proof of the identification theorem. The causal identification argument uses the same  construction as the standard proof for longitudinal modified treatment policies \citep{diaz2023nonparametric}. Fix \(s\leq t\), and a history value
\[
f_s^d=(c_1,a_1^d,\ldots,c_{s-1},a_{s-1}^d,c_s),
\]
where $a_j^d=d_j(a_j, r_j, f_j^d,\epsilon_j)$, and a supported treatment value \(a_s\). Define \(Z_s(a_s,f_s^d)\) as follows. For \(j\leq s\), set \(C_j^{[s]}=c_j\), set \(A_j^{[s]}=a_j^d\) for \(j<s\), and set \(A_s^{[s]}=a_s\). For \(j=s+1,\ldots,t\), recursively define
\begin{align}
C_j^{[s]}
&=f_{C,j}(\bar A_{j-1}^{[s]},\bar C_{j-1}^{[s]},U_{C,j}),
\notag\\
R_j^{[s]}
&=f_{R,j}(\bar A_{j-1}^{[s]},\bar C_j^{[s]},U_{R,j}),
\notag\\
A_j^{\natural,[s]}
&=f_{A,j}(\bar A_{j-1}^{[s]},R_j^{[s]},\bar C_j^{[s]},U_{A,j}),
\notag\\
A_j^{[s]}
&=d_j\{A_j^{\natural,[s]},R_j^{[s]},F_j^{[s],(k)},\epsilon_j\},
\notag
\end{align}
where \(F_j^{[s]}=(C_1^{[s]},A_1^{[s]},\ldots,C_{j-1}^{[s]},A_{j-1}^{[s]},C_j^{[s]})\). Finally, define
\begin{equation}
Z_s(a_s,f_s^d)
:=
h\!\left[f_{C,t+1}\{\bar A_t^{[s]},\bar C_t^{[s]},U_{C,t+1}\}\right].
\label{eq:fh-Z-definition}
\end{equation}
We next state a result that will be used to justify conditioning on the observed visit history.

\begin{lemma}[Exchangeability of visit indicators]\label{lem:fh-past-visit-factorization}
Under \Cref{ass:fh-nuc-markov,ass:fh-measurement-exchange}, for every \(s\leq t\),
\begin{equation}
(U_{R,s},U_{A,s})
\indep
\left(
\underbar U_{C,s+1},
\underbar U_{R,s+1},
\underbar U_{A,s+1}
\right)
\mid J_{s-1}=1,F_s,
\label{eq:fh-past-visit-factorization}
\end{equation}
and
\begin{equation}
C_{s+1}\indep \bar R_s\mid F_s,A_s.
\label{eq:fh-visit-C-independence}
\end{equation}
Consequently, for every finite \(k\), \(C_{s+1}\indep J_s^{(k)}\mid F_s,A_s\).
\end{lemma}

\begin{proof}
Write \(U_s=(U_{R,s},U_{A,s})\) and \(W_{s+1}=(\underbar U_{C,s+1},\underbar U_{R,s+1},\underbar U_{A,s+1})\). We first show by induction that
\begin{equation}
\bar R_{s-1}\indep(U_s,W_{s+1})\mid F_s.
\label{eq:fh-past-vector-independent}
\end{equation}
The result is immediate for \(s=1\). Suppose it holds at time \(s\). Together with \Cref{ass:fh-nuc-markov}, it gives the conditional factorization
\[
p(\bar r_{s-1},u_s,w_{s+1}\mid f_s)
=p(\bar r_{s-1}\mid f_s)p(u_s\mid f_s)p(w_{s+1}\mid f_s).
\]
Under \Cref{ass:fh-measurement-exchange}, conditional on \(F_s\), the pair \((R_s,A_s)\) is a measurable function of \(U_s\). Hence \((\bar R_s,A_s)\indep W_{s+1}\mid F_s\), and therefore \(\bar R_s\indep W_{s+1}\mid F_s,A_s\). Since \(C_{s+1}=f_{C,s+1}(\bar A_s,\bar C_s,U_{C,s+1})\), conditioning additionally on \(C_{s+1}\) leaves \(\bar R_s\) independent of the remaining coordinates of \(W_{s+1}\). Because \(F_{s+1}=(F_s,A_s,C_{s+1})\), this is exactly \eqref{eq:fh-past-vector-independent} at time \(s+1\).

The same factorization shows that conditioning on the event \(J_{s-1}=1\), which is a function of \(\bar R_{s-1}\), does not change the conditional independence of \(U_s\) and \(W_{s+1}\), proving \eqref{eq:fh-past-visit-factorization}. Finally, \((\bar R_s,A_s)\indep W_{s+1}\mid F_s\) implies \(\bar R_s\indep U_{C,s+1}\mid F_s,A_s\). Since \(C_{s+1}\) is a measurable function of \((F_s,A_s,U_{C,s+1})\), \(\bar R_s\indep C_{s+1}\mid F_s,A_s\). This proves \eqref{eq:fh-visit-C-independence}; the stated finite-\(k\) consequence follows because \(J_s^{(k)}\) is a function of \(\bar R_s\).
\end{proof}

\begin{lemma}[Exchangeability of the auxiliary variable]\label{lem:fh-Z-joint-exchange}
Under \Cref{ass:fh-nuc-markov,ass:fh-measurement-exchange}, for every \(s\leq t\), fixed supported \(a_s\), and fixed full-data history \(f_s\),
\begin{equation}
Z_s(a_s,f_s)
\indep
(R_s,A_s)
\mid
J_{s-1}=1,F_s=f_s.
\label{eq:fh-Z-joint-independence}
\end{equation}
Consequently, for every supported \((r_s,a_s')\),
\begin{align}
&\E\{Z_s(a_s,f_s)\mid J_{s-1}=1,R_s=r_s,A_s=a_s',F_s=f_s\}
\nonumber\\
&\qquad=
\E\{Z_s(a_s,f_s)\mid J_s=1,A_s=a_s,F_s=f_s\}.
\label{eq:fh-Z-exchange-expectation}
\end{align}
\end{lemma}

\begin{proof}
For fixed \((a_s,f_s)\), \(Z_s(a_s,f_s)\) is a measurable function of \((\underbar U_{C,s+1},\underbar U_{R,s+1},\underbar U_{A,s+1})\) and the future intervention randomizers \((\epsilon_{s+1},\ldots,\epsilon_t)\). Conditional on \(F_s\), the natural pair \((R_s,A_s)\) is a measurable function of \((U_{R,s},U_{A,s})\). Since the intervention randomizers are independent of all exogenous variables, Equation \eqref{eq:fh-Z-joint-independence} therefore follows from Lemma~\ref{lem:fh-past-visit-factorization}. The first conditional expectation in \eqref{eq:fh-Z-exchange-expectation} equals \(\E\{Z_s(a_s,f_s)\mid J_{s-1}=1,F_s=f_s\}\), and the same is true after conditioning on \((R_s,A_s)=(1,a_s)\). Since \(J_s=J_{s-1}R_s\), the latter event is \(\{J_s=1,A_s=a_s,F_s=f_s\}\). The support condition is supplied by \Cref{ass:fh-support-markov}.
\end{proof}

\begin{lemma}[Recursive continuation of the auxiliary variable]\label{lem:fh-Z-recursion}
Fix \(s=2,\ldots,t\). Suppose the recursion defining \(Z_{s-1}(a_{s-1}^d,f_{s-1}^d)\) generates \(C_s=c_s\), \(R_s=r_s\), and natural treatment \(A_s=a_s\). Let \(f_s^d=(f_{s-1}^d,a_{s-1}^d,c_s)\) and \(a_s^d=d_s\{a_s,r_s,(f_s^d)^{(k)},\epsilon_s\}\). Then
\begin{equation}
Z_{s-1}(a_{s-1}^d,f_{s-1}^d)
=
Z_s(a_s^d,f_s^d).
\label{eq:fh-Z-recursive-consistency}
\end{equation}
\end{lemma}

\begin{proof}
Inside the recursion for \(Z_{s-1}\), the time-\(s\) policy treatment is
\[
A_s^{[s-1]}
=d_s\{A_s^{\natural,[s-1]},R_s^{[s-1]},F_s^{[s-1],(k)},\epsilon_s\}
=d_s\{a_s,r_s,(f_s^d)^{(k)},\epsilon_s\}
=a_s^d.
\]
After this assignment, the time-\((s+1)\) structural equations receive the fixed histories \((\bar A_s^d,\bar C_s)\). By \Cref{ass:fh-measurement-exchange}, \(r_s\) is not an additional argument of those equations. Thus the recursion from time \(s+1\) onward is exactly the recursion defining \(Z_s(a_s^d,f_s^d)\), with the same future exogenous variables, proving \eqref{eq:fh-Z-recursive-consistency}.
\end{proof}

\subsection{Identification through sequential regressions}

\begin{lemma}[Sequential regression formulation of estimand]\label{lem:fh-reference-recursion}
Take \(k=\infty\) for the conditioning history, while allowing the policy itself to depend on its stated \(k\)-lag argument. Define
\begin{align}
\mu_{t+1,t}(a_t,f_t)
&:=\E(Y_{t+1}\mid J_t=1,A_t=a_t,F_t=f_t),
\notag\\
\mu_{t+1,s}(a_s,f_s)
&:=\E\!\left[
\mu_{t+1,s+1}\{d_{s+1}(A_{s+1},R_{s+1},F_{s+1}^{(k)},\epsilon_{s+1}),F_{s+1}\}
\mid J_s=1,A_s=a_s,F_s=f_s
\right]
\label{eq:fh-reference-recursion}
\end{align}
for \(s=t-1,\ldots,1\), where the expectation also averages independently over \(\epsilon_{s+1}\sim Q_{s+1}\). Under \Cref{ass:fh-nuc-markov,ass:fh-measurement-exchange,ass:fh-support-markov},
\begin{equation}
\theta_{t+1}(\bar d_t)
=
\E\!\left[
\mu_{t+1,1}\{d_1(A_1,1,C_1),C_1\}
\right].
\label{eq:fh-reference-target}
\end{equation}
\end{lemma}

\begin{proof}
At baseline \(R_1=1\). For fixed \((a_1,c_1)\), write \(a_1^d=d_1(a_1,1,c_1)\). By the definition of the intervention and \eqref{eq:fh-Z-definition},
\[
Y_{t+1}(\bar d_t)=Z_1\{d_1(A_1,1,C_1),C_1\}.
\]
Therefore
\begin{align}
\theta_{t+1}(\bar d_t)
&=
\int
\E\{Z_1(a_1^d,c_1)\mid A_1=a_1,C_1=c_1\}
\,dP(a_1,c_1)
\nonumber\\
&=
\int
\E\{Z_1(a_1^d,c_1)\mid R_1=1,A_1=a_1^d,C_1=c_1\}
\,dP(a_1,c_1),
\notag
\end{align}
where the second equality is Lemma~\ref{lem:fh-Z-joint-exchange} and \Cref{ass:fh-support-markov}. Consider a generic step \(s=2,\ldots,t\). Suppose the current conditional expectation is
\[
\E\{Z_{s-1}(a_{s-1}^d,f_{s-1}^d)\mid J_{s-1}=1,A_{s-1}=a_{s-1}^d,F_{s-1}=f_{s-1}^d\}.
\]
Applying the tower property to \((C_s,R_s,A_s)\) gives an integral of
\begin{align*}
&\E\{Z_{s-1}(a_{s-1}^d,f_{s-1}^d)
\mid J_{s-1}=1,A_{s-1}=a_{s-1}^d,F_{s-1}=f_{s-1}^d,
C_s=c_s,R_s=r_s,A_s=a_s\}
\end{align*}
with respect to the corresponding conditional law of \((C_s,R_s,A_s)\) and independently over \(\epsilon_s\sim Q_s\). For every value in this integral, set \(f_s^d=(f_{s-1}^d,a_{s-1}^d,c_s)\) and \(a_s^d=d_s\{a_s,r_s,(f_s^d)^{(k)},\epsilon_s\}\). Lemma~\ref{lem:fh-Z-recursion} changes the variable inside the conditional expectation to \(Z_s(a_s^d,f_s^d)\), and Lemma~\ref{lem:fh-Z-joint-exchange} then gives
\begin{align*}
&\E\{Z_s(a_s^d,f_s^d)
\mid J_{s-1}=1,R_s=r_s,A_s=a_s,F_s=f_s^d\}\\
&\qquad=
\E\{Z_s(a_s^d,f_s^d)
\mid J_s=1,A_s=a_s^d,F_s=f_s^d\}.
\end{align*}
Thus the same argument can be repeated at time \(s+1\).

After iterating through time \(t\), the innermost conditional expectation is
\[
\E\{Z_t(a_t^d,f_t^d)\mid J_t=1,A_t=a_t^d,F_t=f_t^d\}.
\]
There is no subsequent treatment modification in \(Z_t\). By structural consistency this equals
\[
\E(Y_{t+1}\mid J_t=1,A_t=a_t^d,F_t=f_t^d)
=\mu_{t+1,t}(a_t^d,f_t^d).
\]
Collapsing the resulting iterated integral backward gives exactly \eqref{eq:fh-reference-recursion} and then \eqref{eq:fh-reference-target}.
\end{proof}

\subsection{Bridge recovery of the reference-state recursion}

\begin{lemma}[One-step bridge identity]\label{lem:fh-one-step-recovery}
Take \(k=\infty\). Fix \(s\leq t\), and write \(\mathcal H_s=\{J_s=1,A_s,F_s\}\). Let \(v^1\) be integrable and let \(v^0\) be integrable and measurable with respect to \((F_s,A_s,C_{s+1},\epsilon_{s+1})\), where \(\epsilon_{s+1}\sim Q_{s+1}\) is independent of the observed-data variables. Then
\begin{align}
&\E\{R_{s+1}v^1+(1-R_{s+1})v^0\mid\mathcal H_s\}
\nonumber\\
&\qquad=
\E\!\left[R_{s+1}\{v^1+(\beta_s^0(F_s,C_{s+1})-1)v^0\}\mid\mathcal H_s\right].
\label{eq:fh-one-step-recovery}
\end{align}
\end{lemma}

\begin{proof}
By iterated expectation, the independence of \(\epsilon_{s+1}\), and the shadow-variable restriction,
\begin{align*}
&\E\{R_{s+1}\beta_s^0(F_s,C_{s+1})v^0\mid\mathcal H_s\}\\
&\quad=
\E\!\left[
\beta_s^0(F_s,C_{s+1})v^0
\P(R_{s+1}=1\mid J_s=1,A_s,F_s,C_{s+1},\epsilon_{s+1})
\mid\mathcal H_s
\right]\\
&\quad=
\E(v^0\mid\mathcal H_s).
\end{align*}
Hence
\[
\E\{(1-R_{s+1})v^0\mid\mathcal H_s\}
=
\E\{R_{s+1}(\beta_s^0-1)v^0\mid\mathcal H_s\}.
\]
Adding \(\E(R_{s+1}v^1\mid\mathcal H_s)\) proves \eqref{eq:fh-one-step-recovery}.
\end{proof}

\begin{lemma}[Identification using the true bridges]\label{lem:fh-true-bridge-anchor}
Under \Cref{ass:fh-nuc-markov,ass:fh-measurement-exchange,ass:fh-known-novisit,ass:fh-support-markov,ass:fh-shadow-markov,ass:fh-positive-markov,ass:fh-markov}, evaluate \eqref{eq:fh-m-recursion}--\eqref{eq:fh-q-recursion} with \(\beta_s=\beta_s^0\) for every \(s\leq t\). Then
\begin{equation}
\theta_{t+1}(\bar d_t)
=
\E\!\left[
 m_{t+1,1}\{d_1(A_1,1,C_1),C_1\}
\right].
\label{eq:fh-true-bridge-anchor}
\end{equation}
\end{lemma}

\begin{proof}
We first take \(k=\infty\). At time \(t\), Lemma~\ref{lem:fh-one-step-recovery} with \(v^0=Y_{t+1}\) gives
\begin{align*}
m_{t+1,t}(a_t,f_t)
&=\E\{R_{t+1}\beta_t^0(F_t,C_{t+1})Y_{t+1}\mid J_t=1,A_t=a_t,F_t=f_t\}\\
&=\E(Y_{t+1}\mid J_t=1,A_t=a_t,F_t=f_t)
=\mu_{t+1,t}(a_t,f_t).
\end{align*}
Suppose \(m_{t+1,s}=\mu_{t+1,s}\) for some \(s\geq2\). Equation \eqref{eq:fh-reference-recursion} at time \(s-1\) can be written
\begin{align*}
&\E\!\left[
R_s\mu_{t+1,s}\{d_s(A_s,1,F_s^{(k)}),F_s\}
+(1-R_s)\mu_{t+1,s}\{t_s(\alpha_s(H_s^{(k)}),\epsilon_s),F_s\}
\mid J_{s-1}=1,A_{s-1},F_{s-1}
\right].
\end{align*}
Indeed, on \(R_s=0\), \Cref{ass:fh-known-novisit} gives \(A_s=\alpha_s(H_s^{(k)})\), and the policy assigns \(t_s\{\alpha_s(H_s^{(k)}),\epsilon_s\}\). The outer expectation also averages independently over \(\epsilon_s\sim Q_s\). Applying Lemma~\ref{lem:fh-one-step-recovery} at time \(s-1\), with \(\mu_{t+1,s}\{t_s(\alpha_s(H_s^{(k)}),\epsilon_s),F_s\}\) as \(v^0\), changes this expression to
\begin{align*}
&\E\!\left[
R_s\left[
\mu_{t+1,s}\{d_s(A_s,1,F_s^{(k)}),F_s\}
\right.\right.\\[-1mm]
&\hspace{22mm}\left.\left.
+\{\beta_{s-1}^0(F_{s-1},C_s)-1\}
\mu_{t+1,s}\{t_s(\alpha_s(H_s^{(k)}),\epsilon_s),F_s\}
\right]
\middle|
\substack{J_{s-1}=1,\ A_{s-1},\\F_{s-1}}
\right].
\end{align*}
By the induction hypothesis and \eqref{eq:fh-q-recursion}, the quantity inside the outer expectation is \(q_{t+1,s}\). Therefore \eqref{eq:fh-m-recursion} gives \(m_{t+1,s-1}=\mu_{t+1,s-1}\). Backward induction and Lemma~\ref{lem:fh-reference-recursion} prove \eqref{eq:fh-true-bridge-anchor} for the complete-history regression state.

For finite \(k\), the same argument applies after reducing the transition law to the retained state. We give the reduction explicitly at one generic time. Lemma~\ref{lem:fh-past-visit-factorization} and \eqref{ass:fh-markov-M1} imply
\begin{equation}
\begin{split}
p(c_{s+1}\mid J_s=1,A_s=a_s,F_s=f_s)
&=p(c_{s+1}\mid A_s=a_s,F_s=f_s)\\
&=p(c_{s+1}\mid A_s=a_s,F_s^{(k)}=f_s^{(k)})\\
&=p(c_{s+1}\mid J_s^{(k)}=1,A_s=a_s,F_s^{(k)}=f_s^{(k)}).
\end{split}
\notag
\end{equation}
The last equality follows by conditioning on the remote history and using \eqref{ass:fh-markov-M1} together with \eqref{eq:fh-visit-C-independence}. \eqref{ass:fh-markov-M3} gives
\begin{align}
&p(r_{s+1}\mid J_s=1,A_s=a_s,F_s=f_s,C_{s+1}=c_{s+1})
\nonumber\\
&\qquad=
p(r_{s+1}\mid J_s^{(k)}=1,A_s=a_s,F_s^{(k)}=f_s^{(k)},C_{s+1}=c_{s+1}),
\notag
\end{align}
and the proof of Lemma~\ref{lem:fh-bridge-equation} shows that \Cref{ass:fh-shadow-markov,ass:fh-positive-markov} reduce to the corresponding \(k\)-lag shadow restriction and positivity. When \(r_{s+1}=1\), \eqref{ass:fh-markov-M2} at time \(s+1\) gives
\begin{align}
&p(a_{s+1}\mid J_{s+1}=1,F_{s+1}=f_{s+1})
\nonumber\\
&\qquad=
p(a_{s+1}\mid J_{s+1}^{(k)}=1,F_{s+1}^{(k)}=f_{s+1}^{(k)}).
\notag
\end{align}
When \(r_{s+1}=0\), the natural treatment is the known value \(\alpha_{s+1}(H_{s+1}^0)\), and the intervention treatment is \(t_{s+1}\{\alpha_{s+1}(H_{s+1}^0),\epsilon_{s+1}\}\). 
\end{proof}

\subsection{Bridge replacement and the adjoint condition}

We first verify the adjoint formula in \eqref{eq:fh-adjoint}. For square-integrable \(h\) and \(u\),
\begin{align*}
\langle T_sh,u\rangle_{A,s}
&=
\E\!\left[
J_s^{(k)}u(A_s,F_s^{(k)})
\E\{R_{s+1}h(F_s^{(k)},C_{s+1})\mid J_s^{(k)}=1,A_s,F_s^{(k)}\}
\right]\\
&=
\E\!\left[
K_{s+1}^{(k)}h(F_s^{(k)},C_{s+1})u(A_s,F_s^{(k)})
\right]\\
&=
\E\!\left[
K_{s+1}^{(k)}h(F_s^{(k)},C_{s+1})
(T_s^\star u)(F_s^{(k)},C_{s+1})
\right]\\
&=\langle h,T_s^\star u\rangle_{C,s}.
\end{align*}

\begin{lemma}[Effect of replacing one bridge solution]\label{lem:fh-bridge-replacement}
Fix \(s\leq t\). Let \(\beta_s\) and \(\widetilde\beta_s\) be two square-integrable solutions of \eqref{eq:fh-bridge-equation}, keep all other bridge functions fixed, and let \((m_{t+1,j},q_{t+1,j})\) and \((\widetilde m_{t+1,j},\widetilde q_{t+1,j})\) be the corresponding recursions. Write \(r_j=m_{t+1,j}-\widetilde m_{t+1,j}\). If \(s<t\), then
\begin{align}
r_s(a_s,f_s)
&=
\E\!\left[
R_{s+1}\{\beta_s(F_s^{(k)},C_{s+1})-\widetilde\beta_s(F_s^{(k)},C_{s+1})\}
\right.\nonumber\\[-1mm]
&\hspace{26mm}\left.
\times
m_{t+1,s+1}\{t_{s+1}(\alpha_{s+1}(H_{s+1}^{(k)}),\epsilon_{s+1}),F_{s+1}^{(k)}\}
\mid J_s^{(k)}=1,A_s=a_s,F_s^{(k)}=f_s
\right].
\label{eq:fh-local-bridge-replacement}
\end{align}
For \(s=t\), the same display holds with the last factor replaced by \(Y_{t+1}\). For \(j=s-1,\ldots,1\),
\begin{align}
r_j(a_j,f_j)
&=
\E\!\left[
R_{j+1}\left[
 r_{j+1}\{d_{j+1}(A_{j+1},1,F_{j+1}^{(k)}),F_{j+1}^{(k)}\}
\right.\right.\nonumber\\[-1mm]
&\hspace{20mm}\left.\left.
+\{\beta_j(F_j^{(k)},C_{j+1})-1\}
 r_{j+1}\{t_{j+1}(\alpha_{j+1}(H_{j+1}^{(k)}),\epsilon_{j+1}),F_{j+1}^{(k)}\}
\right]
\middle|
\substack{J_j^{(k)}=1,\ A_j=a_j,\\F_j^{(k)}=f_j}
\right].
\label{eq:fh-upstream-bridge-replacement}
\end{align}
Consequently,
\begin{align}
&\E\!\left[
 m_{t+1,1}\{d_1(A_1,1,C_1),C_1\}
-\widetilde m_{t+1,1}\{d_1(A_1,1,C_1),C_1\}
\right]
\nonumber\\
&\qquad=
\E\!\left[
\Omega_sR_{s+1}\nu_{t+1,s+1}
\{\beta_s(F_s^{(k)},C_{s+1})-\widetilde\beta_s(F_s^{(k)},C_{s+1})\}
\right]
\nonumber\\
&\qquad=
\E\!\left[
K_{s+1}^{(k)}\phi_{t+1,s}(F_s^{(k)},C_{s+1})
\{\beta_s(F_s^{(k)},C_{s+1})-\widetilde\beta_s(F_s^{(k)},C_{s+1})\}
\right]
\nonumber\\
&\qquad=
\langle\beta_s-\widetilde\beta_s,\phi_{t+1,s}\rangle_{C,s}.
\label{eq:fh-loading-representation}
\end{align}
\end{lemma}

\begin{proof}
The two recursions agree after time \(s\), so \(m_{t+1,s+1}=\widetilde m_{t+1,s+1}\). Subtracting the two definitions of \(q_{t+1,s+1}\) and then applying \eqref{eq:fh-m-recursion} gives \eqref{eq:fh-local-bridge-replacement}; when \(s=t\), subtracting the two terminal pseudo-outcomes gives the stated terminal version. At every earlier time \(j\), the bridge \(\beta_j\) is the same in both recursions. Subtracting \eqref{eq:fh-q-recursion} and then \eqref{eq:fh-m-recursion} therefore gives \eqref{eq:fh-upstream-bridge-replacement}. Applying the density ratios \eqref{eq:fh-omega-t} at each preceding step and iterating to \(j=1\) gives the first equality in \eqref{eq:fh-loading-representation}. The second equality follows from \eqref{eq:fh-phi-loading} and the definition of \(K_{s+1}^{(k)}\), and the last equality is the definition of the \(\langle\cdot,\cdot\rangle_{C,s}\) inner product.
\end{proof}

\subsection{Proof of Theorem~\ref{thm:fh-target-identification}}

\begin{proof}
Lemma~\ref{lem:fh-true-bridge-anchor} proves the theorem when every selected bridge is the true inverse-probability bridge. It remains to show that replacing a selected bridge solution by the true solution does not change the right-hand side of \eqref{eq:fh-theta}. Let
\[
\bar\beta
=
(\beta_1,\ldots,\beta_t)
\qquad\text{and}\qquad
\bar\beta^0
=
(\beta_1^0,\ldots,\beta_t^0).
\]
For \(r=0,\ldots,t\), define the hybrid bridge collection
\[
\bar\beta^{[r]}
=
(\beta_1^0,\ldots,\beta_r^0,
\beta_{r+1},\ldots,\beta_t),
\]
with the conventions
\[
\bar\beta^{[0]}=\bar\beta,
\qquad
\bar\beta^{[t]}=\bar\beta^0.
\]
Every component of every \(\bar\beta^{[r]}\) is a square-integrable solution
of \eqref{eq:fh-bridge-equation}. Let
\((m_{t+1,j}^{[r]},q_{t+1,j}^{[r]})\) denote the recursion
\eqref{eq:fh-terminal-q}--\eqref{eq:fh-q-recursion} evaluated using
\(\bar\beta^{[r]}\), and define
\[
\vartheta_r
:=
\E\!\left[
m_{t+1,1}^{[r]}
\{d_1(A_1,1,C_1),C_1\}
\right].
\]
By Lemma~\ref{lem:fh-true-bridge-anchor},
\begin{equation}
\vartheta_t
=
\theta_{t+1}(\bar d_t),
\label{eq:fh-invariance-true-anchor}
\end{equation}
because \(\bar\beta^{[t]}=\bar\beta^0\). It therefore suffices to show that
\(\vartheta_0=\vartheta_t\). Telescoping over the hybrid bridge collections gives
\begin{equation}
\vartheta_0-\vartheta_t
=
\sum_{s=1}^t
\{\vartheta_{s-1}-\vartheta_s\}.
\label{eq:fh-invariance-telescope}
\end{equation}
We show that every summand on the right-hand side is zero. Fix \(s\leq t\). The bridge collections
\(\bar\beta^{[s-1]}\) and \(\bar\beta^{[s]}\) differ only at time \(s\):
the former uses \(\beta_s\), whereas the latter uses \(\beta_s^0\).
All bridges before time \(s\) are equal to their true values in both
collections, and all bridges after time \(s\) are equal to the originally
selected values in both collections.

This also implies that the loading used to compare these two recursions is
the same under both hybrid collections. Indeed, by
\eqref{eq:fh-Omega}, \(\Omega_s\) depends on the
bridge collection only through
\(\beta_1,\ldots,\beta_{s-1}\), which are identical under
\(\bar\beta^{[s-1]}\) and \(\bar\beta^{[s]}\). Likewise, for \(s<t\),
\(\nu_{t+1,s+1}\) in \eqref{eq:fh-nu} depends on
\(m_{t+1,s+1}\), and the backward recursion defining
\(m_{t+1,s+1}\) involves only the bridges
\(\beta_{s+1},\ldots,\beta_t\). These bridges are also identical under the
two hybrid collections. When \(s=t\),
\(\nu_{t+1,t+1}=Y_{t+1}\) and hence does not depend on any bridge.
Consequently, the quantity
\begin{equation}
\phi_{t+1,s}^{[s]}(F_s^{(k)},C_{s+1})
:=
\E\!\left[
\Omega_s^{[s]}
\nu_{t+1,s+1}^{[s]}
\mid
K_{s+1}^{(k)}=1,F_s^{(k)},C_{s+1}
\right]
\label{eq:fh-invariance-hybrid-loading}
\end{equation}
is the same whether it is evaluated using
\(\bar\beta^{[s-1]}\) or \(\bar\beta^{[s]}\).

Lemma~\ref{lem:fh-bridge-replacement} may therefore be applied to these two
adjacent hybrid recursions, with
\(\widetilde\beta_s=\beta_s^0\). It gives
\begin{equation}
\vartheta_{s-1}-\vartheta_s
=
\left\langle
\beta_s-\beta_s^0,
\phi_{t+1,s}^{[s]}
\right\rangle_{C,s}.
\notag
\end{equation}
Because both \(\beta_s\) and \(\beta_s^0\) solve
\eqref{eq:fh-bridge-equation},
\[
T_s\beta_s=1
\qquad\text{and}\qquad
T_s\beta_s^0=1,
\]
and hence, by linearity of \(T_s\),
\begin{equation}
T_s(\beta_s-\beta_s^0)=0.
\label{eq:fh-null-bridge-difference}
\end{equation}

Moreover, \(\bar\beta^{[s]}\) is itself a square-integrable bridge collection
satisfying \eqref{eq:fh-bridge-equation}. Therefore
\Cref{ass:fh-target-dual-range}, which is imposed for every such bridge
collection, applies to the loading in
\eqref{eq:fh-invariance-hybrid-loading}. Thus there exists a
square-integrable function
\(\lambda_{t+1,s}^{[s]}(A_s,F_s^{(k)})\) such that
\begin{equation}
T_s^\star\lambda_{t+1,s}^{[s]}
=
\phi_{t+1,s}^{[s]}.
\label{eq:fh-invariance-hybrid-adjoint}
\end{equation}
Using \eqref{eq:fh-invariance-hybrid-adjoint}, the adjoint identity, and
\eqref{eq:fh-null-bridge-difference}, we obtain
\begin{align*}
\vartheta_{s-1}-\vartheta_s
&=
\left\langle
\beta_s-\beta_s^0,
\phi_{t+1,s}^{[s]}
\right\rangle_{C,s}\\
&=
\left\langle
\beta_s-\beta_s^0,
T_s^\star\lambda_{t+1,s}^{[s]}
\right\rangle_{C,s}\\
&=
\left\langle
T_s(\beta_s-\beta_s^0),
\lambda_{t+1,s}^{[s]}
\right\rangle_{A,s}\\
&=0.
\end{align*}
Since this holds for every \(s=1,\ldots,t\),
\eqref{eq:fh-invariance-telescope} yields
\[
\vartheta_0=\vartheta_t.
\]
Combining this equality with \eqref{eq:fh-invariance-true-anchor} gives
\[
\E\!\left[
m_{t+1,1}^{[0]}
\{d_1(A_1,1,C_1),C_1\}
\right]
=
\theta_{t+1}(\bar d_t).
\]
Because \(\bar\beta^{[0]}=\bar\beta\) was an arbitrary square-integrable
collection of solutions to \eqref{eq:fh-bridge-equation}, this proves
\eqref{eq:fh-theta} and shows that its right-hand side is invariant to the
selected bridge solutions.
\end{proof}

\subsection{Proof of Theorem~\ref{thm:fh-if}}\label{proof:thm:fh-if}

\begin{assumption}[Regularity of the selected full-history bridge paths]\label{ass:fh-beta-path-regularity}
Along every regular one-dimensional submodel \(\{P_\varepsilon\}\) through \(P\), select for every \(s\leq t\) a square-integrable solution \(\beta_{s,\varepsilon}\) of
\[
T_{s,\varepsilon}\beta_{s,\varepsilon}=1.
\]
Assume there is a square-integrable function
\(\dot\beta_s(F_s^{(k)},C_{s+1})\) such that
\[
\|\beta_{s,\varepsilon}-\beta_s-\varepsilon\dot\beta_s\|_{L^2(P)}
=
o(|\varepsilon|)
\]
for every \(s\leq t\), and assume the domination conditions required to interchange differentiation and the conditional expectations below.
\end{assumption}

\begin{proof}
Fix a regular parametric submodel
\(\{P_\varepsilon:\varepsilon\in(-\delta,\delta)\}\) through \(P=P_0\), and let
\[
S(O)
=
\left.
\frac{\partial}{\partial\varepsilon}
\log p_\varepsilon(O)
\right|_{\varepsilon=0}
\]
denote its score. A dot denotes differentiation at \(\varepsilon=0\). We decompose the derivative into the contribution obtained by holding the selected bridge functions fixed and the contribution induced by their pathwise derivatives.

\medskip
\noindent\textbf{Part I: derivative holding the bridge functions fixed.}

For \(s=1,\ldots,t\), write
\[
e_{t+1,s}(O)
=
\eta_{t+1,s+1}(O)
-
m_{t+1,s}(A_s,F_s^{(k)}).
\]
Holding the bridge functions fixed, differentiation of \eqref{eq:fh-m-recursion} gives
\begin{align}
\dot m_{t+1,s}(a_s,f_s)
={}&
\E\!\left[
e_{t+1,s}(O)S(O)
\mid
J_s^{(k)}=1,A_s=a_s,F_s^{(k)}=f_s
\right]
\nonumber\\
&+
\E\!\left[
\dot\eta_{t+1,s+1}(O)
\mid
J_s^{(k)}=1,A_s=a_s,F_s^{(k)}=f_s
\right].
\label{eq:fh-proof-mdot}
\end{align}
At the terminal step, \(\eta_{t+1,t+1}=R_{t+1}\beta_tY_{t+1}\) is held fixed as a function of the observed data, so
\[
\dot\eta_{t+1,t+1}=0
\]
in this part of the calculation.

Differentiating \eqref{eq:fh-theta} while holding the bridges fixed gives
\begin{align}
\dot\theta^{\mathrm{seq}}
={}&
\E\!\left[
\left\{
m_{t+1,1}\{d_1(A_1,1,C_1),C_1\}
-
\theta_{t+1}(\bar d_t)
\right\}
S(O)
\right]
\nonumber\\
&+
\E\!\left[
\dot m_{t+1,1}\{d_1(A_1,1,C_1),C_1\}
\right].
\notag
\end{align}
By the definition of the pushforward density \(g_1^d\),
\begin{align}
\E\!\left[
\dot m_{t+1,1}\{d_1(A_1,1,C_1),C_1\}
\right]
&=
\E\!\left[
\omega_1^d(A_1,C_1)
\dot m_{t+1,1}(A_1,C_1)
\right]
\nonumber\\
&=
\E\!\left[
\Omega_1
\dot m_{t+1,1}(A_1,C_1)
\right].
\label{eq:fh-proof-baseline-change}
\end{align}
Substituting \eqref{eq:fh-proof-mdot} at \(s=1\) into
\eqref{eq:fh-proof-baseline-change} yields
\begin{align}
\E\!\left[
\Omega_1\dot m_{t+1,1}(A_1,C_1)
\right]
=
\E\!\left[
\Omega_1e_{t+1,1}(O)S(O)
\right]
+
\E\!\left[
\Omega_1\dot\eta_{t+1,2}(O)
\right].
\label{eq:fh-proof-first-step}
\end{align}

For \(s=1,\ldots,t-1\), differentiation of the pseudo-outcome
\eqref{eq:fh-q-recursion}, still holding the bridge functions fixed, gives
\begin{align}
\dot\eta_{t+1,s+1}(O)
=
R_{s+1}\bigg[
&
\dot m_{t+1,s+1}
\{d_{s+1}(A_{s+1},1,F_{s+1}^{(k)}),F_{s+1}^{(k)}\}
\nonumber\\[-1mm]
&+
\{\beta_s(F_s^{(k)},C_{s+1})-1\}
\dot m_{t+1,s+1}
\{t_{s+1}(\alpha_{s+1}(H_{s+1}^{(k)}),\epsilon_{s+1}),F_{s+1}^{(k)}\}
\bigg].
\label{eq:fh-proof-etadot}
\end{align}
The two density-ratio identities in
\eqref{eq:fh-omega-t}, together with
\eqref{ass:fh-markov-M2}, give
\begin{align}
&\E\!\left[
\Omega_s\dot\eta_{t+1,s+1}(O)
\right]
\nonumber\\
&\qquad=
\E\!\left[
\Omega_sR_{s+1}
\left\{
\omega_{s+1}^d(A_{s+1},F_{s+1}^{(k)})
+
\{\beta_s(F_s^{(k)},C_{s+1})-1\}
\omega_{s+1}^t(A_{s+1},F_{s+1}^{(k)})
\right\}
\right.
\nonumber\\[-1mm]
&\hspace{53mm}\left.
\times
\dot m_{t+1,s+1}(A_{s+1},F_{s+1}^{(k)})
\right]
\nonumber\\
&\qquad=
\E\!\left[
\Omega_{s+1}
\dot m_{t+1,s+1}(A_{s+1},F_{s+1}^{(k)})
\right].
\label{eq:fh-proof-propagation}
\end{align}
To justify the first equality, condition on
\(\{J_{s+1}^{(k)}=1,F_{s+1}^{(k)}\}\). On this event, the two evaluations of
\(\dot m_{t+1,s+1}\) in \eqref{eq:fh-proof-etadot} are converted to evaluation at the observed \(A_{s+1}\) by the definitions of \(g_{s+1}^d\) and \(g_{s+1}^t\). Any part of \(\Omega_s\) determined by history preceding the retained state is conditionally independent of \(A_{s+1}\) by \eqref{ass:fh-markov-M2}, so the same change of measure remains valid after multiplication by \(\Omega_s\).

Applying \eqref{eq:fh-proof-mdot} to the right-hand side of
\eqref{eq:fh-proof-propagation} gives
\begin{align*}
&\E\!\left[
\Omega_{s+1}e_{t+1,s+1}(O)S(O)
\right]
+
\E\!\left[
\Omega_{s+1}\dot\eta_{t+1,s+2}(O)
\right].
\end{align*}
Starting from \eqref{eq:fh-proof-first-step}, applying this identity successively through time \(t\), and using
\(\dot\eta_{t+1,t+1}=0\) at the terminal step gives
\begin{align}
\dot\theta^{\mathrm{seq}}
=
\E\!\left[
\left\{
m_{t+1,1}\{d_1(A_1,1,C_1),C_1\}
-
\theta_{t+1}(\bar d_t)
+
\sum_{s=1}^t
\Omega_s e_{t+1,s}(O)
\right\}
S(O)
\right].
\label{eq:fh-proof-seq-gradient}
\end{align}
The expression multiplying \(S(O)\) is exactly
\(D_t^{\mathrm{seq}}(O)\) in \eqref{eq:fh-seq-gradient}. Hence
\[
\dot\theta^{\mathrm{seq}}
=
\E\{D_t^{\mathrm{seq}}(O)S(O)\}.
\]

\medskip
\noindent\textbf{Part II: derivative through the bridge functions.}

We now hold the observed-data law fixed and differentiate the dependence of the recursion on the selected bridge functions. By linearity of the first-order derivative, it is enough to perturb one bridge coordinate at a time and sum the resulting contributions.

Fix \(s\leq t\). The continuation after time \(s\) depends only on
\(\beta_{s+1},\ldots,\beta_t\), while \(\Omega_s\) depends only on
\(\beta_1,\ldots,\beta_{s-1}\). Hence neither changes in the coordinate derivative with respect to \(\beta_s\). If \(s<t\), differentiation of
\eqref{eq:fh-q-recursion} at time \(s+1\) gives
\[
\dot\eta_{t+1,s+1}^{\beta_s}(O)
=
R_{s+1}
\nu_{t+1,s+1}
\dot\beta_s(F_s^{(k)},C_{s+1}),
\]
and the same display holds for \(s=t\) by
\eqref{eq:fh-terminal-q} and the convention
\(\nu_{t+1,t+1}=Y_{t+1}\).
Propagating this perturbation through the preceding regressions by the same calculation as in
\eqref{eq:fh-proof-baseline-change}--\eqref{eq:fh-proof-propagation} gives
\begin{align}
\dot\theta^{\beta_s}
&=
\E\!\left[
\Omega_s
R_{s+1}
\nu_{t+1,s+1}
\dot\beta_s(F_s^{(k)},C_{s+1})
\right]
\nonumber\\
&=
\E\!\left[
K_{s+1}^{(k)}
\phi_{t+1,s}(F_s^{(k)},C_{s+1})
\dot\beta_s(F_s^{(k)},C_{s+1})
\right]
\nonumber\\
&=
\left\langle
\dot\beta_s,
\phi_{t+1,s}
\right\rangle_{C,s}.
\notag
\end{align}
For the second equality, note that \(\Omega_s=0\) whenever
\(J_s^{(k)}=0\), and then apply the definition
\eqref{eq:fh-phi-loading} and iterated expectation.

By \Cref{ass:fh-target-dual-range}, choose
\(\lambda_{t+1,s}\) satisfying
\(T_s^\star\lambda_{t+1,s}=\phi_{t+1,s}\). The adjoint identity gives
\begin{align}
\dot\theta^{\beta_s}
&=
\left\langle
\dot\beta_s,T_s^\star\lambda_{t+1,s}
\right\rangle_{C,s}
\nonumber\\
&=
\left\langle
T_s\dot\beta_s,\lambda_{t+1,s}
\right\rangle_{A,s}
\nonumber\\
&=
\E\!\left[
J_s^{(k)}
\lambda_{t+1,s}(A_s,F_s^{(k)})
R_{s+1}
\dot\beta_s(F_s^{(k)},C_{s+1})
\right].
\label{eq:fh-proof-beta-adjoint}
\end{align}

The bridge equation \eqref{eq:fh-bridge-equation} is equivalent to
\begin{equation}
\E\!\left[
J_s^{(k)}
r_s(A_s,F_s^{(k)})
\{R_{s+1}\beta_s(F_s^{(k)},C_{s+1})-1\}
\right]
=
0
\label{eq:fh-proof-bridge-moment}
\end{equation}
for every square-integrable \(r_s(A_s,F_s^{(k)})\). Differentiating
\eqref{eq:fh-proof-bridge-moment} along the parametric submodel, while holding the test function fixed, gives
\begin{align}
0
={}&
\E\!\left[
J_s^{(k)}
r_s(A_s,F_s^{(k)})
R_{s+1}
\dot\beta_s(F_s^{(k)},C_{s+1})
\right]
\nonumber\\
&+
\E\!\left[
J_s^{(k)}
r_s(A_s,F_s^{(k)})
\{R_{s+1}\beta_s(F_s^{(k)},C_{s+1})-1\}
S(O)
\right].
\label{eq:fh-proof-linearized-bridge}
\end{align}
Taking \(r_s=\lambda_{t+1,s}\) in
\eqref{eq:fh-proof-linearized-bridge} and substituting into
\eqref{eq:fh-proof-beta-adjoint} yields
\begin{align}
\dot\theta^{\beta_s}
=
-\E\!\left[
J_s^{(k)}
\lambda_{t+1,s}(A_s,F_s^{(k)})
\{R_{s+1}\beta_s(F_s^{(k)},C_{s+1})-1\}
S(O)
\right].
\label{eq:fh-proof-beta-gradient}
\end{align}
Summing \eqref{eq:fh-proof-beta-gradient} over
\(s=1,\ldots,t\) gives
\[
\sum_{s=1}^t
\dot\theta^{\beta_s}
=
\E\{D_t^{\mathrm{br}}(O)S(O)\}.
\]
Combining this equality with \eqref{eq:fh-proof-seq-gradient} proves
\[
\left.
\frac{d}{d\varepsilon}
\theta_{t+1}(\bar d_t;P_\varepsilon)
\right|_{\varepsilon=0}
=
\E\{D_t(O)S(O)\}.
\]
Thus \(D_t\) is a gradient.

\end{proof}

\subsection{Proof of Proposition~\ref{prop:fh-second-order}}

\begin{proof}
For a distribution \(X\) and a fixed bridge collection
\(\bar\beta=(\beta_1,\ldots,\beta_t)\), let
\[
\theta_{t+1}(X,\bar\beta)
\]
denote the full-history sequential-regression functional under \(X\), with
the pseudo-outcomes in \eqref{eq:fh-terminal-q}--\eqref{eq:fh-q-recursion}
evaluated using \(\bar\beta\). Let
\(m_{t+1,s,X}^{\bar\beta}\) denote the corresponding sequential
regressions. In particular,
\[
\theta_{t+1}(\bar d_t;X)
=
\theta_{t+1}(X,\bar\beta_X).
\]

We use the same decomposition as in the proof of
Proposition~\ref{prop:vd-second-order}:
\begin{align}
&\theta_{t+1}(P,\bar\beta_P)
-
\theta_{t+1}(F,\bar\beta_F)
\nonumber\\
&\qquad=
\left\{
\theta_{t+1}(P,\bar\beta_P)
-
\theta_{t+1}(P,\bar\beta_F)
\right\}
+
\left\{
\theta_{t+1}(P,\bar\beta_F)
-
\theta_{t+1}(F,\bar\beta_F)
\right\}.
\label{eq:fh-rem-decomposition}
\end{align}
The second term changes the observed-data law while holding the bridge
collection fixed at \(\bar\beta_F\); the first changes only the bridge
collection while holding the law fixed at \(P\).

We first consider the second term. For \(X\in\{P,F\}\), define the
propagation factors associated with the fixed collection \(\bar\beta_F\) by
\[
\Omega_{1,X}^{\bar\beta_F}
=
\omega_{1,X}^d(A_1,C_1)
\]
and, for \(s\geq2\),
\[
\Omega_{s,X}^{\bar\beta_F}
=
\Omega_{s-1,X}^{\bar\beta_F}R_s
\left[
\omega_{s,X}^d(A_s,F_s^{(k)})
+
\{\beta_{s-1,F}(F_{s-1}^{(k)},C_s)-1\}
\omega_{s,X}^t(A_s,F_s^{(k)})
\right].
\]
Thus \(\Omega_{s,F}^{\bar\beta_F}=\Omega_{s,F}\). Applying the same
backward telescoping argument as in
\eqref{eq:vd-rem-standard-seq}--\eqref{eq:vd-rem-standard-seq-remainder},
with the two change-of-measure identities corresponding to
\(\omega_s^d\) and \(\omega_s^t\), gives
\begin{align}
&\theta_{t+1}(P,\bar\beta_F)
-
\theta_{t+1}(F,\bar\beta_F)
\nonumber\\
&\qquad=
\E_P\{D_t^{\mathrm{seq}}(O;F)\}
-
R_{t,2}^{\mathrm{seq}}(P,F;\bar\beta_F),
\label{eq:fh-rem-standard-seq}
\end{align}
where
\begin{align}
R_{t,2}^{\mathrm{seq}}(P,F;\bar\beta_F)
={}&
\E_P\!\left[
\{\omega_{1,P}^d-\omega_{1,F}^d\}
\{m_{t+1,1,F}^{\bar\beta_F}
-m_{t+1,1,P}^{\bar\beta_F}\}
\right]
\nonumber\\
&+
\sum_{s=2}^t
\E_P\!\left[
\Omega_{s-1,F}R_s
\left\{
\omega_{s,P}^d-\omega_{s,F}^d
\right.\right.
\nonumber\\[-1mm]
&\hspace{35mm}\left.\left.
+
(\beta_{s-1,F}-1)
(\omega_{s,P}^t-\omega_{s,F}^t)
\right\}
\right.
\nonumber\\[-1mm]
&\hspace{35mm}\left.
\times
\{m_{t+1,s,F}^{\bar\beta_F}
-m_{t+1,s,P}^{\bar\beta_F}\}
\right].
\notag
\end{align}
All functions in the last display are evaluated at their natural arguments.
The intervention randomizers are averaged over their known laws, as in the
definition of the sequential regressions. By uniform boundedness of the
propagation factors and bridge functions and Cauchy--Schwarz,
\begin{align}
|R_{t,2}^{\mathrm{seq}}(P,F;\bar\beta_F)|
&\lesssim
\sum_{s=1}^t
\|\omega^d_{s,P}-\omega^d_{s,F}\|_{2,P}
\|m^{\bar\beta_F}_{t+1,s,P}
-m^{\bar\beta_F}_{t+1,s,F}\|_{2,P}
\nonumber\\
&\quad+
\sum_{s=2}^t
\|\omega^t_{s,P}-\omega^t_{s,F}\|_{2,P}
\|m^{\bar\beta_F}_{t+1,s,P}
-m^{\bar\beta_F}_{t+1,s,F}\|_{2,P}.
\label{eq:fh-rem-seq-bound}
\end{align}

It remains to study the first term in \eqref{eq:fh-rem-decomposition}. This
is the only additional calculation relative to the observed-history proof,
because the identifying functional now depends on a collection of bridge
functions rather than on a single terminal bridge. Write
\[
\Delta_s
=
\beta_{s,P}-\beta_{s,F},
\qquad s=1,\ldots,t.
\]
For \(s=0,\ldots,t\), define the hybrid bridge collection
\[
\bar\beta^{[s]}
=
(\beta_{1,P},\ldots,\beta_{s,P},
\beta_{s+1,F},\ldots,\beta_{t,F}),
\]
with the conventions
\[
\bar\beta^{[0]}=\bar\beta_F,
\qquad
\bar\beta^{[t]}=\bar\beta_P.
\]
Holding the law fixed at \(P\), set
\[
\vartheta_s
:=
\theta_{t+1}(P,\bar\beta^{[s]}).
\]
The bridge difference therefore has the exact telescoping representation
\begin{align}
&\theta_{t+1}(P,\bar\beta_P)
-
\theta_{t+1}(P,\bar\beta_F)
\nonumber\\
&\qquad=
\vartheta_t-\vartheta_0
=
\sum_{s=1}^t
(\vartheta_s-\vartheta_{s-1}).
\label{eq:fh-rem-bridge-telescope}
\end{align}
Fix \(s\). The adjacent collections \(\bar\beta^{[s]}\) and
\(\bar\beta^{[s-1]}\) differ only in their \(s\)th coordinate. They use
\(\beta_{1,P},\ldots,\beta_{s-1,P}\) at all earlier times and
\(\beta_{s+1,F},\ldots,\beta_{t,F}\) at all later times. Consequently, the
loading in Lemma~\ref{lem:fh-bridge-replacement} is common to the two
recursions: \(\Omega_s\) depends only on bridges before \(s\), whereas
\(\nu_{t+1,s+1}\) depends only on bridges after \(s\). Denote this common
loading, evaluated under \(P\), by
\begin{equation}
\phi_{t+1,s,P}^{[s]}(F_s^{(k)},C_{s+1})
:=
\E_P\!\left[
\Omega_{s,P}\nu_{t+1,s+1,P}^{[s]}
\mid
K_{s+1}^{(k)}=1,F_s^{(k)},C_{s+1}
\right],
\notag
\end{equation}
where \(\nu_{t+1,s+1,P}^{[s]}\) is the continuation obtained from the later
bridges \((\beta_{s+1,F},\ldots,\beta_{t,F})\). Lemma~\ref{lem:fh-bridge-replacement}
now gives
\begin{equation}
\vartheta_s-\vartheta_{s-1}
=
\left\langle
\Delta_s,
\phi_{t+1,s,P}^{[s]}
\right\rangle_{C,s,P}.
\notag
\end{equation}
To separate the first-order contribution from interactions among bridge
errors, add and subtract the loading \(\phi_{t+1,s,P}\) generated by the
full collection \(\bar\beta_P\):
\begin{align}
\vartheta_s-\vartheta_{s-1}
={}&
\left\langle
\Delta_s,
\phi_{t+1,s,P}
\right\rangle_{C,s,P}
\nonumber\\
&+
\left\langle
\Delta_s,
\phi_{t+1,s,P}^{[s]}-\phi_{t+1,s,P}
\right\rangle_{C,s,P}.
\label{eq:fh-rem-hybrid-split}
\end{align}
Summing \eqref{eq:fh-rem-hybrid-split} in
\eqref{eq:fh-rem-bridge-telescope} yields the exact decomposition
\begin{equation}
\theta_{t+1}(P,\bar\beta_P)
-
\theta_{t+1}(P,\bar\beta_F)=
\sum_{s=1}^t
\left\langle
\Delta_s,
\phi_{t+1,s,P}
\right\rangle_{C,s,P}
+
R_{t,2}^{\mathrm{mix}}(P,F),
\label{eq:fh-rem-bridge-expansion}
\end{equation}
where the mixed bridge remainder is explicitly
\begin{equation}
R_{t,2}^{\mathrm{mix}}(P,F)
:=
\sum_{s=1}^t
\left\langle
\Delta_s,
\phi_{t+1,s,P}^{[s]}-\phi_{t+1,s,P}
\right\rangle_{C,s,P}.
\label{eq:fh-rem-mixed-bridge-definition}
\end{equation}
We next show why the last display is second order in the bridge errors. For a
fixed \(s\), let \(\phi_{t+1,s,P}^{[r]}\), \(r=s,\ldots,t\), denote the
time-\(s\) loading under the hybrid collection \(\bar\beta^{[r]}\). Because
the bridges before time \(s\) are equal to their \(P\)-values for every
\(r\geq s\), the propagation factor \(\Omega_{s,P}\) is the same in all of
these loadings. Moreover,
\[
\phi_{t+1,s,P}^{[t]}
=
\phi_{t+1,s,P}.
\]
A second telescope therefore gives
\begin{equation}
\phi_{t+1,s,P}^{[s]}-\phi_{t+1,s,P}
=
-\sum_{r=s+1}^t
\left\{
\phi_{t+1,s,P}^{[r]}
-
\phi_{t+1,s,P}^{[r-1]}
\right\}.
\label{eq:fh-rem-loading-telescope}
\end{equation}
For each \(r>s\), the two loadings in braces differ only because the
continuation after time \(s\) uses \(\beta_{r,P}\) rather than
\(\beta_{r,F}\). Starting from the local perturbation at time \(r\) and
propagating it backward to time \(s+1\) by the same conditional-expectation
and change-of-measure calculation used in
Lemma~\ref{lem:fh-bridge-replacement}, the difference is linear in
\(\Delta_r\). Under the stated uniform boundedness conditions, repeated
conditional expectation is an \(L^2(P)\) contraction and all intervening
multiplicative factors are uniformly bounded. Hence
\begin{equation}
\left\|
\phi_{t+1,s,P}^{[r]}
-
\phi_{t+1,s,P}^{[r-1]}
\right\|_{2,P}
\lesssim
\|\Delta_r\|_{2,P}.
\label{eq:fh-rem-loading-difference-bound}
\end{equation}
Combining \eqref{eq:fh-rem-loading-telescope} and
\eqref{eq:fh-rem-loading-difference-bound}, and then applying
Cauchy--Schwarz to \eqref{eq:fh-rem-mixed-bridge-definition}, gives
\begin{align}
|R_{t,2}^{\mathrm{mix}}(P,F)|
&\lesssim
\sum_{s=1}^t
\|\Delta_s\|_{2,P}
\sum_{r=s+1}^t
\|\Delta_r\|_{2,P}
\nonumber\\
&=
\sum_{1\leq s<r\leq t}
\|\beta_{s,P}-\beta_{s,F}\|_{2,P}
\|\beta_{r,P}-\beta_{r,F}\|_{2,P}.
\label{eq:fh-rem-mixed-bridge-bound}
\end{align}
It remains to identify the first term in
\eqref{eq:fh-rem-bridge-expansion}. By the adjoint range condition applied
to \(\bar\beta_P\), choose \(\lambda_{t+1,s,P}\) satisfying
\[
T_{s,P}^{\star}\lambda_{t+1,s,P}
=
\phi_{t+1,s,P}.
\]
The adjoint identity therefore gives
\begin{align}
\left\langle
\Delta_s,
\phi_{t+1,s,P}
\right\rangle_{C,s,P}
&=
\left\langle
T_{s,P}\Delta_s,
\lambda_{t+1,s,P}
\right\rangle_{A,s,P}
\nonumber\\
&=
\E_P\!\left[
J_s^{(k)}R_{s+1}
\lambda_{t+1,s,P}(A_s,F_s^{(k)})
\Delta_s(F_s^{(k)},C_{s+1})
\right].
\notag
\end{align}
Substituting this expression into \eqref{eq:fh-rem-bridge-expansion} gives
\begin{align}
&\theta_{t+1}(P,\bar\beta_P)
-
\theta_{t+1}(P,\bar\beta_F)
\nonumber\\
&\qquad=
\sum_{s=1}^t
\E_P\!\left[
J_s^{(k)}R_{s+1}
\lambda_{t+1,s,P}
(\beta_{s,P}-\beta_{s,F})
\right]
+
R_{t,2}^{\mathrm{mix}}(P,F).
\label{eq:fh-rem-bridge-P-loading}
\end{align}

We next evaluate the bridge component of the gradient at \(F\), exactly as
in the observed-history proof. Under \(P\), each \(\beta_{s,P}\) satisfies
\(T_{s,P}\beta_{s,P}=1\). Therefore
\begin{align}
\E_P\{D_t^{\mathrm{br}}(O;F)\}
&=
-\sum_{s=1}^t
\E_P\!\left[
J_s^{(k)}
\lambda_{t+1,s,F}
\{R_{s+1}\beta_{s,F}-1\}
\right]
\nonumber\\
&=
\sum_{s=1}^t
\E_P\!\left[
J_s^{(k)}
\lambda_{t+1,s,F}
T_{s,P}(\beta_{s,P}-\beta_{s,F})
\right]
\nonumber\\
&=
\sum_{s=1}^t
\E_P\!\left[
J_s^{(k)}R_{s+1}
\lambda_{t+1,s,F}
(\beta_{s,P}-\beta_{s,F})
\right].
\notag
\end{align}
Subtracting this expression from \eqref{eq:fh-rem-bridge-P-loading} yields
\begin{align}
&\theta_{t+1}(P,\bar\beta_P)
-
\theta_{t+1}(P,\bar\beta_F)
\nonumber\\
&\qquad=
\E_P\{D_t^{\mathrm{br}}(O;F)\}
\nonumber\\
&\qquad\quad+
\sum_{s=1}^t
\E_P\!\left[
J_s^{(k)}R_{s+1}
(\lambda_{t+1,s,P}-\lambda_{t+1,s,F})
(\beta_{s,P}-\beta_{s,F})
\right]
\nonumber\\
&\qquad\quad+
R_{t,2}^{\mathrm{mix}}(P,F).
\label{eq:fh-rem-bridge-exact}
\end{align}
By Cauchy--Schwarz, the middle term is bounded in absolute value by
\begin{equation}
\sum_{s=1}^t
\|\lambda_{t+1,s,P}-\lambda_{t+1,s,F}\|_{2,P}
\|\beta_{s,P}-\beta_{s,F}\|_{2,P}.
\label{eq:fh-rem-bridge-bound}
\end{equation}
Finally, combining \eqref{eq:fh-rem-decomposition},
\eqref{eq:fh-rem-standard-seq}, and \eqref{eq:fh-rem-bridge-exact} gives
\begin{align*}
&\theta_{t+1}(P,\bar\beta_P)
-
\theta_{t+1}(F,\bar\beta_F)
\\
&\qquad=
\E_P\{D_t(O;F)\}
-
R_{t,2}^{\mathrm{seq}}(P,F;\bar\beta_F)
\\
&\qquad\quad+
\sum_{s=1}^t
\E_P\!\left[
J_s^{(k)}R_{s+1}
(\lambda_{t+1,s,P}-\lambda_{t+1,s,F})
(\beta_{s,P}-\beta_{s,F})
\right]
+
R_{t,2}^{\mathrm{mix}}(P,F).
\end{align*}
Equivalently,
\[
\theta_{t+1}(\bar d_t;F)
-
\theta_{t+1}(\bar d_t;P)
=
-\E_P\{D_t(O;F)\}
+
R_{t,2}(P,F),
\]
where, explicitly,
\begin{align*}
R_{t,2}(P,F)
={}&
R_{t,2}^{\mathrm{seq}}(P,F;\bar\beta_F)
\\
&-
\sum_{s=1}^t
\E_P\!\left[
J_s^{(k)}R_{s+1}
(\lambda_{t+1,s,P}-\lambda_{t+1,s,F})
(\beta_{s,P}-\beta_{s,F})
\right]
-
R_{t,2}^{\mathrm{mix}}(P,F).
\end{align*}
The claimed bound follows from \eqref{eq:fh-rem-seq-bound},
\eqref{eq:fh-rem-bridge-bound}, and
\eqref{eq:fh-rem-mixed-bridge-bound}.
\end{proof}